\documentclass[12pt,draftclsnofoot, onecolumn]{IEEEtran}
\usepackage{subfigure}
\usepackage{setspace}
\usepackage{multirow} 
\usepackage{algorithm}
\usepackage{algorithmic}
\usepackage{amsmath}
\usepackage{amssymb}
\usepackage{latexsym}
\usepackage{multirow}
\usepackage{epsfig}
\usepackage{graphics}
\usepackage{graphicx}
\usepackage{mathrsfs}
\usepackage{subfigure}
\usepackage{slashbox}
\usepackage{mathrsfs}
\usepackage{bbding}
\usepackage{cite}
\usepackage{color}
\usepackage{xcolor}
\usepackage{booktabs}
\usepackage{amssymb,mathrsfs,amsmath}
\newcommand{\bm}[1]{\mbox{\boldmath{$#1$}}}
\usepackage{cases}

\usepackage{hyperref}
\usepackage{amsthm}
\theoremstyle{plain}

\newtheorem{theo}{Theorem}
\newtheorem{rem}{Remark}
\newtheorem{lem}{Lemma}

\allowdisplaybreaks [4]

\begin{document}
\begin{sloppypar}	
\title{IRS-Aided Overloaded Multi-Antenna Systems: Joint User Grouping and Resource Allocation} 
\author{Ying~Gao, Qingqing~Wu, Wen~Chen, Yang~Liu,  Ming~Li, and Daniel~Benevides~da Costa \vspace{-1.5cm}
\thanks{Y. Gao is with the Department of Electronic Engineering, Shanghai Jiao Tong University, Shanghai 201210, China, and also with the State Key Laboratory of Internet of Things for Smart City, University of Macau, Macao 999078, China (e-mail: yinggao@um.edu.mo). Q.~Wu and W.~Chen are with the Department of Electronic Engineering, Shanghai Jiao Tong University, Shanghai 201210, China (e-mail: qingqingwu@sjtu.edu.cn; whenchen@sjtu.edu.cn).  Y.~Liu and M.~Li are with the School of Information and Communication Engineering, Dalian University of Technology, Dalian 116024, China (e-mail: yangliu\_613@dlut.edu.cn; mli@dlut.edu.cn). D.~B.~da~Costa is with the Technology Innovation Institute, 9639 Masdar City, Abu Dhabi, United Arab Emirates (email: danielbcosta@ieee.org).}}

\maketitle

\begin{abstract}
	\vspace{-2.5mm}
	This paper studies an intelligent reflecting surface (IRS)-aided multi-antenna simultaneous wireless information and power transfer (SWIPT) system where an $M$-antenna access point (AP) serves $K$ single-antenna information users (IUs) and $J$ single-antenna energy users (EUs) with the aid of an IRS with phase errors. We explicitly concentrate on overloaded scenarios where $K + J > M$ and $K \geq M$. 
	Our goal is to maximize the minimum throughput among all the IUs by optimizing the allocation of resources (including time, transmit beamforming at the AP, and reflect beamforming at the IRS), while guaranteeing the minimum amount of harvested energy at each EU. Towards this goal, we propose two user grouping (UG) schemes, namely, the non-overlapping UG scheme and the overlapping UG scheme, where the difference lies in whether identical IUs can exist in multiple groups. Different IU groups are served in orthogonal time dimensions, while the IUs in the same group are served simultaneously with all the EUs via spatial multiplexing. The two problems corresponding to the two UG schemes are mixed-integer non-convex optimization problems and difficult to solve optimally. We first provide a method to check the feasibility of these two problems, and then propose efficient algorithms for them based on the big-M formulation, the penalty method, the block coordinate descent, and the successive convex approximation. Simulation results show that: 1) the non-robust counterparts of the proposed robust designs are unsuitable for practical IRS-aided SWIPT systems with phase errors since the energy harvesting constraints cannot be satisfied; 2) the proposed UG strategies can significantly improve the max-min throughput over the benchmark schemes without UG or adopting random UG; 3) the overlapping UG scheme performs much better than its non-overlapping counterpart when the absolute difference between $K$ and $M$ is small and the EH constraints are not stringent. \vspace{-2mm}  
\end{abstract}

\begin{IEEEkeywords}
	\vspace{-3mm}
    Intelligent reflecting surface, overloaded multi-antenna systems, SWIPT, user grouping, phase errors.
\end{IEEEkeywords}

\section{Introduction}
\vspace{-1mm}
Radio-frequency (RF) signals-enabled wireless power transfer (WPT) has been recognized as a viable and convenient solution for providing virtually perpetual energy supplies to wireless devices \cite{2015_Suzhi_WPT_intro}. Moreover, since RF signals carry both energy and information, the integration of WPT and wireless information transmission (WIT) spurs a new paradigm, namely, \emph{simultaneous wireless information and power transfer} (SWIPT), which has drawn an upsurge of interest \cite{2008_Varshney_SWIPT,2013_Rui_MIMO_SWIPT}. However, as the path loss is proportional to the
transmission distance, the performance of SWIPT systems is basically limited by the low efficiency and short range of WPT. Although using massive antenna arrays at the transmitter can overcome this issue, the required high energy consumption and hardware cost hinder its practical implementation, which calls for an energy-efficient and cost-effective alternative solution \cite{2017_Qingqing_overview_5G}. 

Recently, intelligent reflecting surface (IRS) has been proposed as a promising solution that can improve the spectral efficiency and/or energy efficiency of various wireless systems \cite{2020_Qingqing_IRS_Intro}. Specifically, an IRS is a planar array consisting of a substantial quantity of low-cost passive metamaterial elements, each of which can be adapted to tune the phase shifts of the incoming signals, enabling the reconfiguration of the wireless propagation environment for boosting the efficiencies of WPT and WIT \cite{2019_Qingqing_Joint,2022_Qingqing_WEIT_overview,2022_Qingqing_WPCN}. Furthermore, IRSs possess several other attractive benefits, including a compact form factor, lightweight construction, and conformal geometry. Therefore, IRSs can be mounted on surfaces of arbitrary shapes, accommodating diverse application scenarios \cite{2019_Marco_intro}. Inspired by these advantages, several works have investigated the integration of IRSs into SWIPT systems, e.g., \cite{2020_Qingqing_SWIPT_letter,2020_Qingqing_SWIPT_QoS,2020_Cunhua_SWIPT,2020_Wei_SWIPT_secure,2021_Shayan_SWIPT,2022_Yang_SWIPT,2023_Ying_IFC}. Two distinct research lines can be identified depending on whether the information users (IUs) and energy users (EUs) are geographically separated or co-located. For the case of separated IUs and EUs, the authors of \cite{2020_Qingqing_SWIPT_letter} studied the joint design of the transmit precoder at the access point (AP) and the phase shifts at the IRS for maximizing the weighted sum-power of the EUs in an IRS-aided multiple-input single-output (MISO) SWIPT system. Their simulation results demonstrated that the IRS can significantly improve the power harvested by the EUs in its vicinity and enlarge the signal-to-interference-plus-noise ratio (SINR)-energy region. Moreover, adopting the same system model as in \cite{2020_Qingqing_SWIPT_letter}, the authors of \cite{2020_Qingqing_SWIPT_QoS} investigated the transmit power minimization problem. Also, the weighted sum-rate of the IUs was maximized in \cite{2020_Cunhua_SWIPT} for an IRS-assisted multiple-input multiple-output (MIMO) SWIPT system. 
On the other hand, for the case of co-located users with both information decoding and energy harvesting (EH) requirements, the authors of \cite{2021_Shayan_SWIPT} considered the power splitting (PS) receiver structure and maximized the minimum energy efficiency among the users to guarantee user fairness in a MISO SWIPT system aided by an IRS. Additionally, the rate-energy (R-E) trade-off of a single user employing either PS or time switching (TS) receiver structures in an IRS-aided multicarrier MISO SWIPT system was studied in \cite{2022_Yang_SWIPT}. 

All the aforementioned works assumed that the phase shifts induced by the IRS reflecting elements can be estimated perfectly and/or set precisely to the desired values, which, however, may be ideal due to the intrinsic hardware imperfection of IRSs. The phase shift deviations from the desired values caused by imperfect phase estimation and/or low-precision phase configuration are referred to as \emph{phase errors} \cite{2019_Badiu_PSE}. Several studies on wireless communication systems aided by IRSs with phase errors have been carried out, e.g., \cite{2021_Xing_error,2022_Zheng_PSE_WPCN,2022_Anastasios_PSE,2021_Tianxiong_PSE,2022_Zaid_PSE}. These works indicate that if ignoring the phase errors at the design stage, then the system performance would degrade since the system resources are not utilized properly. Among them, there are two commonly used distributions for modeling the phase errors, i.e., the uniform distribution and the Von Mises distribution. For the former case, the authors of \cite{2021_Xing_error} derived a closed-form expression for the average rate of an IRS-aided SISO system. In \cite{2022_Zheng_PSE_WPCN}, the sum throughput was maximized for an IRS-aided multiuser SISO wireless powered communication network (WPCN). 
For the latter case, the outage probability of an IRS-aided SISO system was analyzed in \cite{2021_Tianxiong_PSE}. Also, the authors of \cite{2022_Zaid_PSE} explored the performance of a double-IRS-assisted multiuser MISO system over spatially correlated channels. However, to the best of our knowledge, the research on IRS-aided SWIPT systems in the presence of phase errors is still in its infancy. If the design parameters are determined without considering the phase errors, the systems employing them may fail to meet the quality-of-service (QoS) requirements at the IUs and the EUs, and also cannot utilize the resources properly to maximize the system performance. Hence, it is necessary and important to take the phase errors into account in practical IRS-aided SWIPT systems. 

In addition to the above restriction, prior works on IRS-aided SWIPT systems (e.g., \cite{2020_Qingqing_SWIPT_letter,2020_Qingqing_SWIPT_QoS,2020_Cunhua_SWIPT,2020_Wei_SWIPT_secure,2021_Shayan_SWIPT}) have the following limitation. To be specific, in \cite{2020_Qingqing_SWIPT_letter,2020_Qingqing_SWIPT_QoS,2020_Cunhua_SWIPT,2020_Wei_SWIPT_secure,2021_Shayan_SWIPT}, the transmitter sends information and energy simultaneously via spatial multiplexing to all the IUs and EUs over the whole transmission interval. While this transmission strategy can neutralize multiuser interference and guarantee user fairness when the number of transmit antennas is sufficient, it fails to achieve satisfactory results in overloaded scenarios where the number of IUs and EUs is large such that the number of signals multiplexed in the spatial domain exceeds the number of transmit antennas, even with the aid of IRSs. Since overloaded scenarios are gaining increasing importance with the ever-growing demands for ultra-high connectivity, it is necessary to pay attention to them \cite{2017_Hamdi_overloaded}. Then, a question arises: \emph{how to improve the minimum throughput performance among the IUs in overloaded scenarios}? Intuitively, the fewer the number of IUs served by the transmitter with the help of IRSs over the given frequency band, the higher the achievable SINR of each IU. Inspired by this, user grouping (UG) can be pursued, where different IU groups are served in orthogonal time dimensions to avoid inter-group interference, and all the EUs can still harvest energy over the whole transmission duration. Although a higher SINR can be achieved per IU in this case, the duration that the transmitter serves each IU is reduced. Thus, it is unknown whether the max-min throughput performance can be improved or not by doing so. If the answer is yes, then another question arises: \emph{does allowing overlap among the IU groups lead to more significant performance improvement}? This question is motivated by the fact that as a super-scheme of non-overlapping UG, overlapping UG offers a better utilization of the system resources. For overlapping UG, the IUs that belong to multiple groups can benefit from an extended duration of service compared to when they exist in only one group. Nevertheless, as the number of IUs within a single group increases, the achievable SINR of each IU in the group decreases. Hence, it is unclear whether and when overlapping UG can noticeably outperform non-overlapping UG. The answer to this question can offer important engineering insights. For instance, considering that non-overlapping UG is easier to implement, if it exhibits comparable performance to overlapping UG, then it is undoubtedly a better choice for practical systems. 
Finally, since the spatial correlation among the IUs in the same group significantly impacts the system performance and can be changed by IRSs, both the non-overlapping and overlapping UG schemes should be carefully designed. \looseness=-1

Motivated by these considerations, this paper investigates an IRS-aided overloaded SWIPT system which is composed of an IRS with phase errors, an AP with $M$ antennas, and two sets of single-antenna users, i.e., $K$ IUs and $J$ EUs. In addition, $K + J > M$ and $K \geq M$. We aim at maximizing the minimum throughput among all the IUs via optimizing the allocation of resources (including time, transmit beamforming at the AP, and IRS phase shifts), subject to the EH requirements of the EUs. Our main contributions are summarized as follows. \looseness=-1
\begin{itemize}
	\item Unlike existing works (e.g., \cite{2020_Qingqing_SWIPT_letter,2020_Qingqing_SWIPT_QoS,2020_Cunhua_SWIPT,2020_Wei_SWIPT_secure,2021_Shayan_SWIPT}) where all the IUs are served simultaneously, we propose two UG schemes, namely, the non-overlapping UG scheme and the overlapping UG scheme, to assign the IUs into several groups. The second scheme is a super-scheme of the first one, distinguishing itself by allowing each IU to be assigned into multiple groups. 
	The transmission time is divided into several time slots, each for one group. In each time slot, the IUs in the corresponding group are served simultaneously with all the EUs via spatial multiplexing. We formulate two max-min throughput maximization problems corresponding to the two UG schemes, denoted by (P1) and (P2), respectively. These two problems are mixed-integer non-convex optimization problems, which are much more challenging to solve than those in \cite{2020_Qingqing_SWIPT_letter,2020_Qingqing_SWIPT_QoS,2020_Cunhua_SWIPT,2020_Wei_SWIPT_secure,2021_Shayan_SWIPT} that do not involve UG-related binary optimization variables.  \looseness=-1
	\item For (P1) and (P2), we first provide a method to check their feasibility. Then, we propose a computationally efficient algorithm to solve (P1) suboptimally by applying the proper change of variables, the big-M formulation, the penalty method, the block coordinate descent (BCD), and the successive convex approximation (SCA). To proceed, we prove that removing the UG-related binary variables in (P2) does not compromise optimality, which reveals that although (P2) is a general case of (P1), it is easier to solve. Due to the similarity between (P1) and the simplified version of (P2) (denoted by (P2')), the algorithm proposed for (P1) is modified to find a suboptimal solution of (P2') (and thus (P2)). 
	\item Numerical results verify the effectiveness of our proposed algorithms and indicate the importance of robust design for practical IRS-aided SWIPT systems with phase errors since a non-robust design ignoring the phase errors generally leads to an infeasible EH solution. Furthermore, our proposed UG strategies can achieve remarkable improvements in max-min throughput compared to the cases without UG or adopting random UG. 
	In addition, the overlapping UG scheme is  preferable for scenarios where the absolute difference between $K$ and $M$ is small and the EH constraints are loose, since it significantly surpasses the non-overlapping UG scheme in these scenarios. By contrast, the non-overlapping UG scheme is a more favorable choice for the opposite scenarios, because it performs comparably to the overlapping UG scheme in these scenarios and is easier to implement in practice.   
\end{itemize} 

The remainder of this paper is organized as follows. Section \ref{Sec:model_formu} elaborates on the system model and problem formulations for an IRS-aided overloaded SWIPT system under two different UG strategies. Section \ref{Sec:feasi_check} provides a feasibility checking method for the formulated problems. In Section \ref{Sec:P1_solution} and \ref{Sec:P2_solution}, we propose computationally efficient algorithms to solve the formulated problems suboptimally. In Section \ref{Sec:simulation}, we evaluate the performance of our proposed algorithms via simulations. Finally, Section \ref{Sec:conclusion} concludes the paper.  

\emph{Notations:} $\mathbb C$ denotes the complex space. $\mathbb C^{M\times N}$ represents the space of $M\times N$ complex-valued matrices. Denote by $\mathbb H^M$ the set of all $M$-dimensional complex Hermitian matrices. $\mathbf 0$ and $\mathbf I$ are an all-zero matrix and an identity matrix, respectively, whose dimensions are determined by the context. For a square matrix $\mathbf S$, $\mathbf S\succeq \mathbf 0$ means that $\mathbf S$ is positive semidefinite while ${\rm tr}\left(\mathbf S\right)$ denotes its trace. For two square matrices $\mathbf S_1$ and $\mathbf S_2$, $\mathbf S_1 \succeq \mathbf S_2$ $(\mathbf S_1 \preceq \mathbf S_2)$ indicates that $\mathbf S_1 - \mathbf S_2$ is positive (negative) semidefinite. $\left\|\cdot\right\|_2$ stands for the maximum
singular value of a matrix. Let ${\rm rank}(\cdot)$ be the rank of a matrix. We denote the conjugate transpose and expectation operators by $(\cdot)^H$ and $\mathbb E\left(\cdot\right)$, respectively. $\left\|\cdot\right\|$ and $[\cdot]_i$ represent the Euclidean norm and the $i$-th element of a vector, respectively. ${\rm diag}\left(\cdot\right)$ denotes the diagonalization operation. $\mathcal{CN}\left(\mathbf x,\mathbf \Sigma\right)$ represents a complex Gaussian distribution with a mean vector $\mathbf x$ and co-variance matrix $\mathbf \Sigma$. For a scalar $x$, $\left|x\right|$ denotes its modulus. For a set $\mathcal X$, $\left|\mathcal X\right|$ denotes its cardinality. $\jmath \triangleq \sqrt{-1}$ refers to the imaginary unit. Denote ${\rm Re}\{\cdot\}$ as the real part of a complex number. $\odot$ denotes the Hadamard product. 

\begin{figure}[!t]
	\vspace{-2mm}
	\centering
	\subfigure[]{\label{fig:system_model_a}
		\includegraphics[scale=0.68]{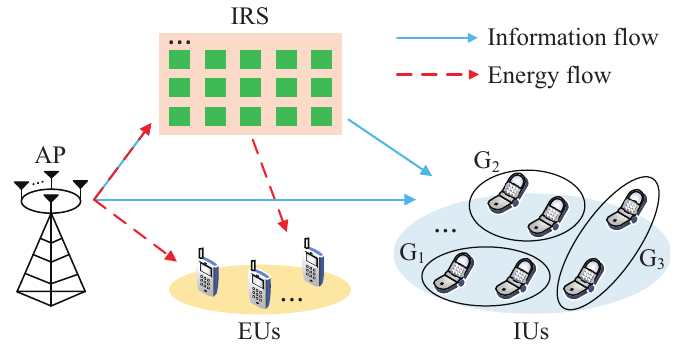}}
	\hspace{2mm}
	\subfigure[]{\label{fig:system_model_b}
		\includegraphics[scale=0.68]{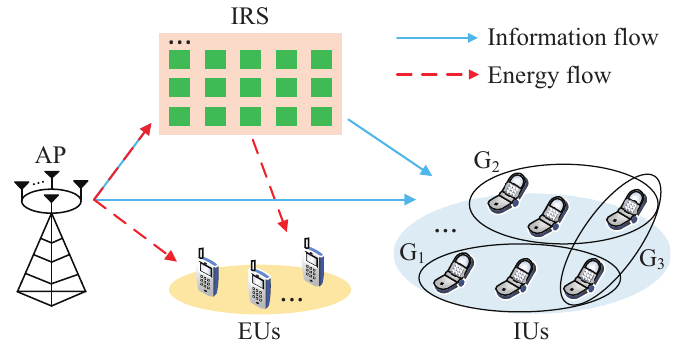}}
	\caption{Illustration of an IRS-aided SWIPT system with different UG strategies: (a) Non-overlapping UG; (b) Overlapping UG. }
	\label{fig:system_model}
	\vspace{-3mm}
\end{figure}

\section{System Model and Problem Formulation}\label{Sec:model_formu}
\subsection{System Model}

This paper considers an IRS-aided overloaded multiuser MISO downlink SWIPT system consisting of an $N$-element passive IRS, an $M$-antenna AP, $K$ single-antenna IUs, and $J$ single-antenna EUs, where $K + J > M$ and $K \geq M$. The sets of reflecting elements, IUs, and EUs are denoted by $\mathcal N$, $\mathcal K$, and $\mathcal J$, respectively, with $\left|\mathcal N\right| = N$, $\left|\mathcal K\right| = K$, and $\left|\mathcal J\right| = J$. It is assumed that the $K$ IUs can be assigned into at most $L$ groups, indexed by G$_1$, $\cdots$, G$_L$. Define a binary variable $a_{k,\ell}$, $k\in\mathcal K, \ell \in \mathcal L \triangleq \{1,\cdots,L\}$, which indicates that the $k$-th IU is assigned into the $\ell$-th group if $a_{k,\ell} = 1$; otherwise, $a_{k,\ell} = 0$. As illustrated in Fig. \ref{fig:system_model}, we consider two UG schemes, i.e., the non-overlapping UG scheme and the overlapping UG scheme, according to whether there are identical IUs in different groups. For the non-overlapping scheme, we have $\sum_{\ell \in\mathcal L} a_{k,\ell} \leq 1$, whereas for the overlapping scheme, there is no constraint on the value of $\sum_{\ell \in\mathcal L} a_{k,\ell}$, $ \forall k\in\mathcal K$. Furthermore, the total transmission time $T$ is divided into $L$ time slots, each occupying a duration of $\tau_\ell \geq 0$ ($\ell\in\mathcal L$), satisfying $\sum_{\ell \in\mathcal L}\tau_\ell \leq T$. In time slot $\ell$, the AP transmits energy and information simultaneously to all the EUs and only the IUs in G$_\ell$ over the given frequency band, as shown in Fig. \ref{fig:frame_structure}. By relying on linear precoding, the complex baseband transmitted signal from the AP at time slot $\ell$, $\ell\in\mathcal L$, can be expressed as $\bm x_\ell = \sum_{k\in\mathcal K}a_{k,\ell}\mathbf w_{k,\ell}s_{k} + \bm x_{\mathrm E,\ell}$, 
where $s_k \in\mathbb C$ denotes the transmitted data symbol for IU $k$, which is precoded by the precoding vector $\mathbf w_{k,\ell} \in\mathbb C^{M\times 1}$ at time slot $\ell$ if $a_{k,\ell} = 1$. Suppose that $s_k \sim \mathcal CN\left(0,1\right)$, $\forall k\in\mathcal K$ and $\{s_k\}$ are independent over $k$. In addition, $\bm x_{\mathrm E,\ell} \in\mathbb C^{M\times 1}$ denotes the transmitted energy signal at time slot $\ell$ with covariance matrix $\mathbf W_{\mathrm E,\ell} = \mathbb E\left(\bm x_{\mathrm E,\ell}\bm x^H_{\mathrm E,\ell}\right) \succeq \mathbf 0$, and the rank of $\mathbf W_{\mathrm E,\ell}$ determines the number of energy beams that are spatially transmitted \cite{2014_Jie_WPT}.    

\begin{figure}[!t]
	\centering
	\includegraphics[width=0.8\textwidth]{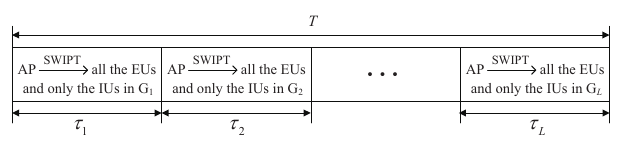}\vspace{-1.5mm}
	\caption{Illustration of the transmission protocol.}
	\label{fig:frame_structure}
	\vspace{-3mm}
\end{figure}

The quasi-static flat-fading model is assumed for all the channels. Let $\mathbf F \in\mathbb C^{N\times M}$, $\mathbf h_{d,k}^H \in\mathbb C^{1\times M}$, $\mathbf g_{d,j}^H \in\mathbb C^{1\times M}$, $\mathbf h_{r,k}^H \in\mathbb C^{1\times N}$, and $\mathbf g_{r,j}^H \in\mathbb C^{1\times N}$ denote the channel coefficients from the AP to the IRS, from the AP to IU $k$, from the AP to EU $j$, from the IRS to IU $k$, and from the IRS to EU $j$, respectively. The cascaded channels from the AP to IU $k$ and EU $j$ via the IRS can be denoted as $\mathbf \Phi_k = {\rm diag}\left(\mathbf h_{r,k}^H\right)\mathbf F$ and $\mathbf \Psi_j = {\rm diag}\left(\mathbf g_{r,j}^H\right)\mathbf F$, respectively. We assume that the perfect channel state information of the direct and cascaded channels can be acquired using existing channel estimation methods such as \cite{2020_G.T._Channel_PARAFAC-Based, 2021_Wei_Channel_RIS-Empowered}. Besides, denoted by $\mathbf \Theta_\ell = {\rm diag}\left( e^{\jmath \theta_{\ell,1}}, \cdots, e^{\jmath \theta_{\ell,N}}\right)$ and $\tilde{\mathbf \Theta}_\ell = {\rm diag}\left( e^{\jmath \tilde{\theta}_{\ell,1}}, \cdots, e^{\jmath \tilde{\theta}_{\ell,N}}\right)$ the phase-shift matrix and the phase-error matrix at the IRS at time slot $\ell$, respectively, where $\theta_{\ell,n} \in [0,2\pi)$ stands for the phase shift induced by the $n$-th element and $\tilde{\theta}_{\ell,n}$ represents the additive random phase error that reflects the imperfection in phase estimation and/or phase configuration. Moreover, $\tilde{\theta}_{\ell,n}$ is assumed to be uniformly distributed on $\left[-{\pi}/{2}, {\pi}/{2}\right]$, $\forall \ell \in\mathcal L, n\in\mathcal N$ \cite{2021_Xing_error}. Then, the received signal at IU $k$ at time slot $\ell$ is given by 
\begin{align}
y_{k,\ell}^{\rm I} & = \left(\bm h_{r,k}^H\mathbf \Theta_\ell\tilde{\mathbf \Theta}_{\ell}\mathbf F + \mathbf h_{d,k}^H \right) \bm x_\ell + n_k = \left(\mathbf v_\ell\odot\tilde {\mathbf v}_\ell\right)^H\mathbf H_k\bm x_\ell + n_k, \ k\in\mathcal K, \ell\in\mathcal L, 
\end{align}
where $\mathbf v_\ell = \left[\mathbf u_\ell; 1\right]$ with $\mathbf u_\ell = \left[e^{\jmath\theta_{\ell,1}},\cdots, e^{\jmath\theta_{\ell,N}}\right]^H$,  $\tilde {\mathbf v}_\ell = \left[\tilde {\mathbf u}_\ell; 1\right]$ with $\tilde {\mathbf u}_\ell = \left[e^{\jmath\tilde {\theta}_{\ell,1}},\cdots, e^{\jmath\tilde {\theta}_{\ell,N}}\right]^H$, $\mathbf H_k = \left[\mathbf \Phi_k;\mathbf h_{d,k}^H\right]$, and $n_k \sim\mathcal {CN} \left(0, \sigma_k^2 \right)$ represents the additive white Gaussian noise with variance $\sigma_k^2$ at IU $k$. Assuming that the IUs cannot cancel the interference caused by the energy signals, the SINR of IU $k$ at time slot $\ell$ can be written as
\begin{align}\label{equ:SINR}
\gamma_{k,\ell} = \frac{a_{k,\ell}\left|\left( \mathbf v_\ell\odot\tilde {\mathbf v}_\ell\right)^H\mathbf H_k\mathbf w_{k,\ell}\right|^2}{\sum_{i\in \mathcal K\backslash\{k\}}a_{i,\ell}\left|\left( \mathbf v_\ell\odot\tilde {\mathbf v}_\ell\right)^H\mathbf H_k\mathbf w_{i,\ell}\right|^2 + {\rm tr}\left(\mathbf H_k^H\left( \mathbf v_\ell\odot\tilde {\mathbf v}_\ell\right)\left( \mathbf v_\ell\odot\tilde {\mathbf v}_\ell\right)^H\mathbf H_k\mathbf W_{\mathrm E,\ell}\right) + \sigma_k^2}, 
\end{align} 
On the other hand, by adopting the widely used linear EH model \cite{2020_Qingqing_SWIPT_letter,2020_Qingqing_SWIPT_QoS,2020_Cunhua_SWIPT}, the harvested RF-band energy at EU $j$, $j\in\mathcal J$, over the whole transmission duration can be expressed as 
\begin{align}\label{equ:E}
Q_j = \sum_{\ell\in\mathcal L}\tau_\ell\left(\sum_{k\in\mathcal K}a_{k,\ell}\left|\left( \mathbf v_\ell\odot\tilde {\mathbf v}_\ell\right)^H\mathbf G_j\mathbf w_{k,\ell} \right|^2 + {\rm tr}\left(\mathbf G_j^H\left( \mathbf v_\ell\odot\tilde {\mathbf v}_\ell\right)\left( \mathbf v_\ell\odot\tilde {\mathbf v}_\ell\right)^H\mathbf G_j\mathbf W_{\mathrm E, \ell} \right) \right), 
\end{align}
where $\mathbf G_j = \left[\mathbf \Psi_j;\mathbf g_{d,j}^H\right]$ and the negligible noise power is ignored.

Note that $\gamma_{k,\ell}$ and $Q_j$ contain the random phase errors that are generally unknown. In view of this, we consider the expectations of them.  

\begin{theo}\label{theo}
	\vspace{-3mm}
	The expectations of $\gamma_{k,\ell}$ and $Q_j$ are respectively given by
	{\addtolength{\jot}{5pt}
	\begin{align}
	\mathbb E_{\tilde{\mathbf v}_\ell}\{ \gamma_{k,\ell}\} & = \frac{a_{k,\ell}\mathbf w_{k,\ell}^H\mathbf X_{k,\ell}\mathbf w_{k,\ell}}{\sum_{i\in \mathcal K\backslash\{k\}}a_{i,\ell}\mathbf w_{i,\ell}^H\mathbf X_{k,\ell}\mathbf w_{i,\ell} + {\rm tr}\left(\mathbf X_{k,\ell}\mathbf W_{\mathrm E,\ell}\right)+ \sigma_k^2} \triangleq \hat \gamma_{k,\ell}, \  k\in\mathcal K, \ell\in\mathcal L, \label{expecta_SINR}\\ 
	\mathbb E_{\tilde{\mathbf v}_\ell}\{ Q_j\} & = \sum_{\ell\in\mathcal L}\tau_\ell\left(\sum_{k\in\mathcal K}a_{k,\ell}\mathbf w_{k,\ell}^H\mathbf Y_{j,\ell}\mathbf w_{k,\ell} + {\rm tr}\left(\mathbf Y_{j,\ell}\mathbf W_{\mathrm E, \ell} \right) \right), \  j\in\mathcal J, \label{expecta_E}
	\end{align}}%
	where $\mathbf X_{k,\ell} = \mathbf H_k^H{\rm diag}\left(\mathbf v_\ell\right)\mathbf Z{\rm diag}\left(\mathbf v_{\ell}^H\right)\mathbf H_k$, $\mathbf Y_{j,\ell} = \mathbf G_j^H{\rm diag}\left(\mathbf v_\ell\right)\mathbf Z{\rm diag}\left(\mathbf v_{\ell}^H\right)\mathbf G_j$, and 
	\begin{align}
	\mathbf Z = \begin{bmatrix}
	1& \frac{4}{\pi^2} & \cdots & \frac{4}{\pi^2} & \frac{2}{\pi} \\
	\frac{4}{\pi^2} & 1 & \cdots & \frac{4}{\pi^2}  & \frac{2}{\pi} \\
	\vdots & \vdots & \ddots & \vdots & \vdots \\
	\frac{4}{\pi^2} & \frac{4}{\pi^2} & \cdots &  1 & \frac{2}{\pi} \\
	\frac{2}{\pi} & \frac{2}{\pi} & \cdots & \frac{2}{\pi} & 1 \\
	\end{bmatrix} \in\mathbb R^{\left(N+1\right) \times \left(N+1\right)}.
	\end{align} 
\end{theo}
\begin{proof}
Please refer to Appendix \ref{Appen_A}.	
\end{proof}
  
\subsection{Problem Formulation}
In this paper, we aim to maximize the minimum throughput among all the IUs, denoted by $\eta \triangleq \min_{k\in\mathcal K} \sum_{\ell\in\mathcal L}\tau_\ell\log_2\left(1 + \hat\gamma_{k,\ell} \right)$, by jointly optimizing the UG variables $\{a_{k,\ell}\}$, the time allocation $\{\tau_\ell\}$, the information precoders $\{\mathbf w_{k,\ell}\}$ and the energy covariance matrices $\{\mathbf W_{\mathrm E,\ell}\}$ at the AP, and the IRS phase-shift vectors $\{\mathbf v_\ell\}$ while satisfying the EH constraints at the EUs. For the non-overlapping UG scheme, we can formulate the problem of interest as follows
\begin{subequations}
	\setlength\abovedisplayskip{5pt}
	\setlength\belowdisplayskip{5pt}
	\begin{eqnarray}
	\hspace{-6mm}\text{(P1)}: & \underset{\substack{\eta, \{\mathbf w_{k,\ell}\}, \{\mathbf W_{\mathrm E,\ell} \succeq \mathbf 0\}, \\\{ a_{k,\ell}\}, \{\tau_\ell\},\{\mathbf v_\ell\}} }{\max}&  \eta\\
    &\text{s.t.} & \hspace{-8mm} \sum_{\ell\in\mathcal L}\tau_\ell\log_2\left(1 + \hat\gamma_{k,\ell} \right) \geq \eta, \ \forall k\in\mathcal K, \label{P1_cons:b}\\ 
    && \hspace{-8mm}\sum_{\ell\in\mathcal L}\tau_\ell\left(\sum_{k\in\mathcal K}a_{k,\ell}\mathbf w_{k,\ell}^H\mathbf Y_{j,\ell}\mathbf w_{k,\ell} + {\rm tr}\left(\mathbf Y_{j,\ell}\mathbf W_{\mathrm E, \ell} \right) \right) \geq E, \ \forall j\in\mathcal J, \label{P1_cons:c}\\
	&& \hspace{-8mm} \sum_{k\in\mathcal K}a_{k,\ell}\left\|\mathbf w_{k,\ell} \right\|^2 + {\rm tr}\left(\mathbf W_{\mathrm E,\ell} \right)  \leq P, \ \forall \ell\in\mathcal L, \label{P1_cons:d}\\
	&& \hspace{-8mm} \sum_{\ell\in\mathcal L} \tau_\ell \leq T, \ \tau_\ell \geq 0, \ \forall \ell\in\mathcal L, \label{P1_cons:e}\\
	&& \hspace{-8mm} a_{k,\ell} \in\{0,1\}, \ \forall k\in\mathcal K, \ell\in\mathcal L,\label{P1_cons:f}\\
	&& \hspace{-8mm} \sum_{\ell \in\mathcal L} a_{k,\ell} \leq 1, \ \forall k\in\mathcal K, \label{P1_cons:g}\\
	&& \hspace{-8mm} \left|\left[\mathbf v_\ell\right]_n\right| = 1, \ \left[\mathbf v_\ell\right]_{N+1} = 1, \ \forall \ell\in\mathcal L, n\in\mathcal N, \label{P1_cons:h}  
	\end{eqnarray}
\end{subequations}
where constraint \eqref{P1_cons:c} indicates that each EU is required to harvest at least $E$ Joule (J) energy and constraint \eqref{P1_cons:d} implies that the AP's instantaneous transmit power cannot exceed $P$. Similarly, the minimum throughput maximization problem corresponding to the overlapping UG scheme can be formulated as
\begin{subequations}
	\setlength\abovedisplayskip{5pt}
	\setlength\belowdisplayskip{5pt}
	\begin{eqnarray}
	\hspace{-6mm}\text{(P2)}: &\underset{\substack{\eta, \{\mathbf w_{k,\ell}\}, \{\mathbf W_{\mathrm E,\ell} \succeq \mathbf 0\}, \\\{ a_{k,\ell}\}, \{\tau_\ell\},\{\mathbf v_\ell\}} }{\max}&  \eta\\
	&\text{s.t.}& \hspace{-8mm} \eqref{P1_cons:b} - \eqref{P1_cons:f}, \eqref{P1_cons:h}. 
	\end{eqnarray}
\end{subequations} 
Note that the only difference between (P1) and (P2) is that (P1) includes an extra constraint \eqref{P1_cons:g}. Both (P1) and (P2) are challenging to solve for the following reasons: 1) the variables $\{a_{k,\ell}\}$ are binary, making \eqref{P1_cons:b}-\eqref{P1_cons:d} involve integer constraints; 2) even with fixed $\{a_{k,\ell}\}$, \eqref{P1_cons:b}-\eqref{P1_cons:d} are non-convex constraints due to the coupling of all other variables; 3) the unit-modulus constraints on the IRS phase shifts in \eqref{P1_cons:h} are non-convex. As a result, (P1) and (P2) are both mixed-integer non-convex optimization problems, which are typically NP-hard and non-trivial to solve optimally. \looseness=-1

\section{Feasibility Checking for (P1) and (P2)}\label{Sec:feasi_check}
Prior to solving (P1) and (P2), we first check their feasibility, i.e., whether the EH requirement of each EU can be satisfied under the given AP's transmit power and transmission duration. To this end, we define $\delta \triangleq \min_{j\in\mathcal J}\sum_{\ell\in\mathcal L}\tau_\ell{\rm tr}\left(\mathbf Y_{j,\ell}\mathbf W_{\mathrm E,\ell} \right)$ and consider the following minimum harvested energy maximization problem: 
\begin{subequations}\label{prob:feasi_E}
	\setlength\abovedisplayskip{5pt}
	\setlength\belowdisplayskip{5pt}
	\begin{eqnarray}
	 &\underset{\delta, \{\mathbf W_{\mathrm E,\ell} \succeq \mathbf 0\}, \{\tau_\ell\},\{\mathbf v_\ell\}}{\max}& \delta\\
	&\text{s.t.}& \hspace{-9mm}  \sum_{\ell\in\mathcal L}\tau_\ell{\rm tr}\left(\mathbf Y_{j,\ell}\mathbf W_{\mathrm E,\ell} \right) \geq \delta, \ \forall j\in\mathcal J, \label{E_cons:b}\\
	&& \hspace{-9mm} {\rm tr}\left(\mathbf W_{\mathrm E,\ell} \right) \leq P, \ \forall \ell\in\mathcal L,\label{E_cons:c}\\
	&& \hspace{-9mm}\sum_{\ell\in\mathcal L} \tau_\ell \leq T, \ \tau_\ell \geq 0, \ \forall \ell\in\mathcal L, \label{E_cons:d}\\
	&& \hspace{-9mm} \left| \left[\mathbf v_\ell\right]_n\right| \leq 1, \ \left[\mathbf v_\ell\right]_{N+1} = 1, \ \forall \ell\in\mathcal L, n\in\mathcal N,  \label{E_cons:e}
	\end{eqnarray}
\end{subequations} 
which is non-convex because the optimization variables are strongly coupled in constraint \eqref{E_cons:b}. Given that it is difficult, if not impossible, to solve this problem directly, we alternately solve its subproblems concerning different sets of variables based on the principle of BCD \cite{2001_Tseng_BCD}, as detailed in the following. 

\subsection{Optimizing $\left\lbrace \{\mathbf W_{\mathrm E,\ell}\}, \{\tau_\ell\} \right\rbrace$ for Given $\{\mathbf v_\ell\}$} With given $\{\mathbf v_\ell\}$, by applying the change of variables $\mathbf S_{\mathrm E,\ell} = \tau_\ell\mathbf W_{\mathrm E,\ell}$, $\forall \ell\in\mathcal L$, the subproblem with respect to (w.r.t.) $\left\lbrace \{\mathbf W_{\mathrm E,\ell}\}, \{\tau_\ell\} \right\rbrace$ can be equivalently expressed as  
\begin{subequations}\label{E_sub1}
    \setlength\abovedisplayskip{5pt}
	\setlength\belowdisplayskip{5pt}
	\begin{eqnarray}
	&\underset{\delta, \{\mathbf S_{\mathrm E,\ell} \succeq \mathbf 0\}, \{\tau_\ell\}}{\max}& \delta\\
	&\text{s.t.}& \hspace{-2mm}  \sum_{\ell\in\mathcal L}{\rm tr}\left(\mathbf Y_{j,\ell}\mathbf S_{\mathrm E,\ell} \right) \geq \delta, \ \forall j\in\mathcal J, \label{E_sub1_cons:b}\\
	&& \hspace{-2mm} {\rm tr}\left(\mathbf S_{\mathrm E,\ell} \right)  \leq \tau_{\ell}P, \ \forall \ell\in\mathcal L,\label{E_sub1_cons:c}\\
	&& \hspace{-2mm} \eqref{E_cons:d}.
	\end{eqnarray}
\end{subequations} 
By direct inspection, problem \eqref{E_sub1} is a convex  semidefinite program (SDP), and its optimal solution, denoted by $\{\{\mathbf S_{\mathrm E,\ell}^\star\},\{\tau_\ell^\star\}\}$, can be found by ready-made solvers, e.g., CVX \cite{2004_S.Boyd_cvx}. Moreover, the optimal original variables $\{\mathbf W_{\mathrm E,\ell}^\star\}$ can be recovered from $\{\{\mathbf S_{\mathrm E,\ell}^\star\},\{\tau_\ell^\star\}\}$ by setting $\mathbf W_{\mathrm E,\ell}^\star = \frac{\mathbf S_{\mathrm E,\ell}^\star}{\tau_\ell^\star}$ if $\tau_\ell^\star > 0$ and $\mathbf W_{\mathrm E,\ell}^\star = \mathbf 0$ otherwise, $\forall \ell\in\mathcal L$.    
  
\subsection{Optimizing $\{\mathbf v_\ell\}$ for Given $\left\lbrace \{\mathbf W_{\mathrm E,\ell}\},\{\tau_\ell\} \right\rbrace$}
For any given $\left\lbrace \{\mathbf W_{\mathrm E,\ell}\},\{\tau_\ell\} \right\rbrace$, the subproblem of problem \eqref{prob:feasi_E} for optimizing $\{\mathbf v_\ell\}$ can be written as 
{\setlength\abovedisplayskip{5pt}
\begin{align}\label{E_sub2}
\underset{\delta,\{\mathbf v_\ell\}}{\max} \hspace{3mm} \delta \hspace{7mm}
\text{s.t.} \hspace{3mm}\eqref{E_cons:b}, \eqref{E_cons:e}.
\end{align}}%
It is hard to state whether constraint \eqref{E_cons:b} is convex since the optimization variables $\{\mathbf v_\ell\}$ are not exposed in the current form of \eqref{E_cons:b}. To tackle this issue, we introduce the following lemma.   
\begin{lem}\label{lem}
	\vspace{-3mm}
	Constraint \eqref{E_cons:b} can be equivalently converted to
	{\setlength\abovedisplayskip{5pt}
	\setlength\belowdisplayskip{5pt}
	\begin{align}\label{E_cons:b_eqv}
	\sum_{\ell\in\mathcal L'}\tau_\ell\mathbf v_{\ell}^H\mathbf Q_{j,\mathrm E, \ell}\mathbf v_{\ell} \geq \delta, \ \forall j\in\mathcal J, 
	\end{align}}%
	where $\mathcal L' = \{\ell|\tau_\ell > 0\} \subseteq \mathcal L$ and $\mathbf Q_{j,\mathrm E,\ell} = \sum_{m=1}^{r_{\rm E,\ell}}q_{\ell,m}{\rm diag}\left(\mathbf G_j\mathbf w_{\rm E,\ell,m}\right)\mathbf Z\left( {\rm diag}\left(\mathbf G_j\mathbf w_{\rm E,\ell,m}\right)\right)^H$ with $r_{\rm E,\ell} = {\rm rank}\left(\mathbf W_{\mathrm E,\ell}\right) \geq 1$, $q_{\ell,1}, \cdots, q_{\ell,r_{\rm E,\ell}}$ denoting the eigenvalues of $\mathbf W_{\rm E, \ell}$, and $\mathbf w_{\rm E,\ell,m}$ being the  unit-norm eigenvector of $\mathbf W_{\rm E, \ell}$ corresponding to $q_{\ell,m}$, $m\in\{1,\cdots,r_{\rm E,\ell}\}$.   
\end{lem}
\begin{proof}
\vspace{-3mm}
Please refer to Appendix \ref{Appen_B}.
\vspace{-2mm}
\end{proof}
Note that constraint \eqref{E_cons:b_eqv} is in the form of a super-level set of convex quadratic functions, which makes it non-convex but allows the application of the iterative SCA technique \cite{2010_Dinh_SCA_converge}. Specifically, given the local feasible point $\mathbf v_{\ell}^t$ in the $t$-th iteration of SCA, we can replace the convex term $\mathbf v_{\ell}^H\mathbf Q_{j,\mathrm E, \ell}\mathbf v_{\ell}$ with its first-order Taylor expansion-based lower bound, yielding a convex subset of constraint \eqref{E_cons:b_eqv} expressed as 
{\setlength\abovedisplayskip{5pt}
\setlength\belowdisplayskip{5pt}
\begin{align}\label{E_cons:b_eqv_sca}
\sum_{\ell\in\mathcal L'}\tau_\ell\left( 2{\rm Re}\left\lbrace\mathbf v_\ell^H\mathbf Q_{j,\mathrm E,\ell}\mathbf v_\ell^t\right\rbrace - \big(\mathbf v_\ell^t\big)^H\mathbf Q_{j,\mathrm E,\ell}\mathbf v_\ell^t\right) \geq \delta, \ \forall j\in\mathcal J. 
\end{align}}%
As a result, the optimization problem to be solved in the $t$-th iteration of SCA is given by
{\setlength\abovedisplayskip{5pt}
\setlength\belowdisplayskip{5pt} 
\begin{align}\label{E_sub2_sca}
\underset{\delta,\{\mathbf v_\ell\}}{\max} \hspace{3mm} \delta \hspace{7mm}
\text{s.t.} \hspace{3mm}\eqref{E_cons:b_eqv_sca}, \eqref{E_cons:e},
\end{align}}%
which is a convex quadratically constrained quadratic program (QCQP) and thus can be optimally solved by existing solvers such as CVX \cite{2004_S.Boyd_cvx}. In addition, the optimal $\{\mathbf v_\ell^\star\}$ must satisfy $\left| \left[\mathbf v_\ell^\star\right]_n\right| = 1$, $\forall \ell\in\mathcal L', n\in\mathcal N$, for achieving maximum signal reflection. By iteratively solving problem \eqref{E_sub2_sca} until convergence is reached, we can obtain a locally optimal solution of problem \eqref{E_sub2} \cite{2010_Dinh_SCA_converge}. 

\subsection{Overall Algorithm}
In summary, the proposed algorithm updates $\left\lbrace \{\mathbf W_{\mathrm E,\ell}\}, \{\tau_\ell\} \right\rbrace$ and $\{\mathbf v_\ell\}$ in an alternating manner. The computational complexity of updating $\left\lbrace \{\mathbf W_{\mathrm E,\ell}\}, \{\tau_\ell\} \right\rbrace$ via solving problem \eqref{E_sub1} is $\mathcal O\left(\sqrt{M}\log_2\left({1}/{\varepsilon}\right)\left(\beta M^3 + \beta^2 M^2 + \beta^3\right)\right)$ \cite{2010_Imre_SDR_complexity} with $\beta \triangleq J + L$ and $\varepsilon$ denoting the prescribed accuracy, and that of updating $\{\mathbf v_\ell\}$ via iteratively solving problem \eqref{E_sub2_sca} until SCA converges is $\mathcal O\left(I_0\sqrt{LN + 2J}\log_2\left({1}/{\varepsilon}\right)N^3L^3J \right)$ \cite{2014_K.wang_complexity} with $I_0$ representing the required number of SCA iterations. This algorithm is guaranteed to converge since the objective value of problem \eqref{prob:feasi_E} is non-decreasing with the update iteration index and has a finite upper bound. Moreover, any limit point of the BCD procedure is a stationary point of problem \eqref{prob:feasi_E} \cite{2001_Tseng_BCD}. Once the objective value exceeds $E$ in the BCD procedure, we can stop the iterations and verify that (P1) and (P2) are feasible. For another case where the proposed algorithm converges with an objective value less than $E$, we consider (P1) and (P2) to be infeasible. 
 
\section{Proposed Algorithm for (P1)}\label{Sec:P1_solution}
In this section, we aim to solve (P1). First of all, we deal with the non-convex unit-modulus constraints in \eqref{P1_cons:h} by relaxing them to $\left| \left[\mathbf v_\ell\right]_n\right| \leq 1$, $\forall \ell\in\mathcal L, n\in\mathcal N$. As such, an upper bound of the optimal value of (P1) can be obtained by solving the following problem
\begin{subequations}\label{P1_relaxed}
	\setlength\abovedisplayskip{5pt}
	\setlength\belowdisplayskip{5pt}
	\begin{eqnarray}
	&\underset{\substack{\eta, \{\mathbf w_{k,\ell}\}, \{\mathbf W_{\mathrm E,\ell} \succeq \mathbf 0\}, \\\{ a_{k,\ell}\}, \{\tau_\ell\},\{\mathbf v_\ell\}} }{\max}&  \eta\\
	&\text{s.t.}&  \hspace{-8mm} \eqref{P1_cons:b}-\eqref{P1_cons:g}, \\
	&& \hspace{-8mm} \left|\left[\mathbf v_\ell\right]_n\right| \leq 1, \ \left[\mathbf v_\ell\right]_{N+1} = 1, \ \forall \ell\in\mathcal L, n\in\mathcal N.  \label{P1_cons:h_relaxed}
	\end{eqnarray}
\end{subequations}
To facilitate the solution of problem \eqref{P1_relaxed}, we define $\mathbf W_{k,\ell} = \mathbf w_{k,\ell}\mathbf w_{k,\ell}^H$, satisfying $\mathbf W_{k,\ell}\succeq \mathbf 0$ and ${\rm rank}\left(\mathbf W_{k,\ell}\right) \leq 1$, $\forall k\in\mathcal K$, $\ell\in\mathcal L$. Then, constraints \eqref{P1_cons:b}-\eqref{P1_cons:d} can be converted to
{\setlength\abovedisplayskip{6pt}
\setlength\belowdisplayskip{6pt}
\begin{align}
& \sum_{\ell\in\mathcal L}\tau_\ell\log_2\left(1 + \frac{a_{k,\ell}{\rm tr}\left(\mathbf X_{k,\ell}\mathbf W_{k,\ell}\right)}{\sum_{i\in \mathcal K\backslash\{k\}}a_{i,\ell}{\rm tr}\left(\mathbf X_{k,\ell}\mathbf W_{i,\ell}\right)  + {\rm tr}\left(\mathbf X_{k,\ell}\mathbf W_{\mathrm E,\ell}\right)+ \sigma_k^2} \right) \geq \eta, \ \forall k\in\mathcal K, \label{P1_cons:b_eqv}\\ 
& \sum_{\ell\in\mathcal L}\tau_\ell\left(\sum_{k\in\mathcal K}a_{k,\ell}{\rm tr}\left(\mathbf Y_{j,\ell}\mathbf W_{k,\ell}\right) + {\rm tr}\left(\mathbf Y_{j,\ell}\mathbf W_{\mathrm E, \ell} \right) \right) \geq E, \ \forall j\in\mathcal J, \label{P1_cons:c_eqv}\\
&  \sum_{k\in\mathcal K}a_{k,\ell}{\rm tr}\left( \mathbf W_{k,\ell}\right)  + {\rm tr}\left(\mathbf W_{\mathrm E,\ell} \right) \leq P, \ \forall \ell\in\mathcal L. \label{P1_cons:d_eqv}
\end{align}}%
By applying the change of variables $\mathbf S_{k,\ell} = \tau_\ell\mathbf W_{k,\ell}$, $\forall k\in\mathcal K$, $\ell\in\mathcal L$ and recalling the variables $\{\mathbf S_{\mathrm E, \ell}\}$ defined in the previous section, we can further transform constraints \eqref{P1_cons:b_eqv}-\eqref{P1_cons:d_eqv} into 
{\setlength\abovedisplayskip{7pt}
\setlength\belowdisplayskip{6pt}  
\begin{align}
& \sum_{\ell\in\mathcal L}\tau_\ell\log_2\left(1 + \frac{\frac{a_{k,\ell}{\rm tr}\left(\mathbf X_{k,\ell}\mathbf S_{k,\ell}\right)}{\tau_\ell}}{\frac{\sum_{i\in \mathcal K\backslash\{k\}}a_{i,\ell}{\rm tr}\left(\mathbf X_{k,\ell}\mathbf S_{i,\ell}\right)}{\tau_\ell} + \frac{{\rm tr}\left(\mathbf X_{k,\ell}\mathbf S_{\mathrm E,\ell}\right)}{\tau_\ell} + \sigma_k^2} \right) \geq \eta, \ \forall k\in\mathcal K, \label{P1_cons:b_eqv2}\\ 
& \sum_{\ell\in\mathcal L}\left(\sum_{k\in\mathcal K}a_{k,\ell}{\rm tr}\left(\mathbf Y_{j,\ell}\mathbf S_{k,\ell}\right) + {\rm tr}\left(\mathbf Y_{j,\ell}\mathbf S_{\mathrm E, \ell} \right) \right) \geq E, \ \forall j\in\mathcal J, \label{P1_cons:c_eqv2}\\
&  \sum_{k\in\mathcal K}a_{k,\ell}{\rm tr}\left( \mathbf S_{k,\ell}\right)  + {\rm tr}\left(\mathbf S_{\mathrm E,\ell} \right)  \leq \tau_\ell P, \ \forall \ell\in\mathcal L, \label{P1_cons:d_eqv2}
\end{align}}%
with $\mathbf S_{k,\ell}\succeq \mathbf 0$, ${\rm rank}\left(\mathbf S_{k,\ell}\right) \leq 1$, and $\mathbf S_{\mathrm E, \ell}\succeq \mathbf 0$, $\forall k\in\mathcal K, \ell\in\mathcal L$. Next, the big-M formulation \cite{2016_Kwan_big-M} is adopted to tackle the coupling between the binary variables $\{a_{k,\ell}\}$ and the continuous variables $\{\mathbf S_{k,\ell}\}$ in \eqref{P1_cons:b_eqv2}-\eqref{P1_cons:d_eqv2}. Specifically, we introduce auxiliary variables $\tilde{\mathbf S}_{k,\ell} = a_{k,\ell}\mathbf S_{k,\ell}$, $\forall k\in\mathcal K$, $\ell\in\mathcal L$, and impose the following additional constraints: 
\begin{subequations}\label{big_M_cons}
	\begin{align}
	& \tilde{\mathbf S}_{k,\ell} \preceq a_{k,\ell}PT\mathbf I, \ \forall k\in\mathcal K, \ell \in\mathcal L, \label{big_M_cons:a}\\
	& \tilde{\mathbf S}_{k,\ell} \preceq \mathbf S_{k,\ell}, \ \tilde{\mathbf S}_{k,\ell} \succeq \mathbf 0, \ \forall k\in\mathcal K, \ell \in\mathcal L, \label{big_M_cons:b}\\
	& \tilde{\mathbf S}_{k,\ell} \succeq \mathbf S_{k,\ell} - (1 - a_{k,\ell})PT\mathbf I, \ \forall k\in\mathcal K, \ell \in\mathcal L, \label{big_M_cons:c}\\
	& {\rm rank}\left(\tilde{\mathbf S}_{k,\ell}\right) \leq 1, \ \forall k\in\mathcal K, \ell \in\mathcal L. \label{big_M_cons:d}
	\end{align}
\end{subequations}
It can be verified that when the constraints in \eqref{P1_cons:f} and \eqref{big_M_cons} are satisfied, constraints \eqref{P1_cons:b_eqv2}-\eqref{P1_cons:d_eqv2} are respectively equivalent to
\begin{align}
& \sum_{\ell\in\mathcal L}\tau_\ell\log_2\left(1 + \frac{\frac{{\rm tr}\left(\mathbf X_{k,\ell}\tilde{\mathbf S}_{k,\ell}\right)}{\tau_\ell}}{\frac{\sum_{i\in \mathcal K\backslash\{k\}}{\rm tr}\left(\mathbf X_{k,\ell}\tilde{\mathbf S}_{i,\ell}\right)}{\tau_\ell} + \frac{{\rm tr}\left(\mathbf X_{k,\ell}\mathbf S_{\mathrm E,\ell}\right)}{\tau_\ell} + \sigma_k^2} \right) \geq \eta, \ \forall k\in\mathcal K, \label{P1_cons:b_eqv3}\\ 
& \sum_{\ell\in\mathcal L}\left(\sum_{k\in\mathcal K}{\rm tr}\left(\mathbf Y_{j,\ell}\tilde{\mathbf S}_{k,\ell}\right) + {\rm tr}\left(\mathbf Y_{j,\ell}\mathbf S_{\mathrm E, \ell} \right) \right) \geq E, \ \forall j\in\mathcal J, \label{P1_cons:c_eqv3}\\
& \sum_{k\in\mathcal K}{\rm tr}\left(\tilde{\mathbf S}_{k,\ell}\right)  + {\rm tr}\left(\mathbf S_{\mathrm E,\ell} \right) \leq \tau_\ell P, \ \forall \ell\in\mathcal L. \label{P1_cons:d_eqv3}
\end{align}
Based on the above results, by replacing constraints \eqref{P1_cons:b}-\eqref{P1_cons:d} in problem \eqref{P1_relaxed} with \eqref{P1_cons:b_eqv3}-\eqref{P1_cons:d_eqv3} and taking \eqref{big_M_cons} into account, we can rewrite problem \eqref{P1_relaxed} in its equivalent form, as follows  
\begin{align}\label{P1_relaxed_eqv}
\underset{\eta, \mathcal Z}{\max} \hspace{3mm} \eta \hspace{8mm}
\text{s.t.} \hspace{3mm}  \eqref{P1_cons:e}-\eqref{P1_cons:g}, \eqref{P1_cons:h_relaxed}, \eqref{big_M_cons}-\eqref{P1_cons:d_eqv3}, 
\end{align}
where $\mathcal Z \triangleq \left\lbrace \left\lbrace \tilde{\mathbf S}_{k,\ell} \in \mathbb H^M\right\rbrace, \{\mathbf S_{k,\ell} \in \mathbb H^M\}, \{\mathbf S_{\mathrm E,\ell} \succeq \mathbf 0\}, \{ a_{k,\ell}\}, \{\tau_\ell\},\{\mathbf v_\ell\}\right\rbrace$. Since the binary constraint \eqref{P1_cons:f} is an obstacle to solving problem \eqref{P1_relaxed_eqv}, we equivalently re-express it as \looseness=-1
\begin{subequations}
	\begin{align}
	& 0 \leq a_{k,\ell} \leq 1, \ \forall k\in\mathcal K, \ell \in\mathcal L, \label{P1_cons:f_1}\\
	& a_{k,\ell} - a_{k,\ell}^2 \leq 0, \ \forall k\in\mathcal K, \ell\in\mathcal L. \label{P1_cons:f_2}
	\end{align}
\end{subequations}
Note that \eqref{P1_cons:f_1} is a linear constraint while \eqref{P1_cons:f_2} is a reverse convex constraint that yields a disconnected feasible region. To handle \eqref{P1_cons:f_2}, we incorporate it into the objective function of problem \eqref{P1_relaxed_eqv} via a multiplicative penalty function based on the penalty method \cite{2012_Le_penalty}, yielding the following problem
\begin{align}\label{P1_relaxed_eqv2}
\underset{\eta,\mathcal Z}{\max} \hspace{3mm} \eta - \rho h\big(\{a_{k,\ell}\}\big) \hspace{7mm}
\text{s.t.} \hspace{3mm}  \eqref{P1_cons:e}, \eqref{P1_cons:g}, \eqref{P1_cons:h_relaxed}, \eqref{big_M_cons}-\eqref{P1_cons:d_eqv3}, \eqref{P1_cons:f_1}, 
\end{align} 
where $h\big(\{a_{k,\ell}\}\big) \triangleq \sum_{\ell \in\mathcal L}\sum_{k \in\mathcal K} \left(a_{k,\ell} - a_{k,\ell}^2\right)$ and $\rho > 0$ serves as a penalty parameter to penalize the violation of constraint \eqref{P1_cons:f_2}. Notably, to maximize the objective function of problem \eqref{P1_relaxed_eqv2} when $\rho\rightarrow \infty$, the optimal $\{a_{k,\ell}^\star\}$ should meet the condition $h\big(\{a_{k,\ell}^\star\}\big)\leq 0$. On the other hand, since $\{a_{k,\ell}^\star\}$ satisfy constraint \eqref{P1_cons:f_1}, we have $h\big(\{a_{k,\ell}^\star\}\big) \geq 0$. Thus, $h\big(\{a_{k,\ell}^\star\}\big) = 0$ and accordingly $a_{k,\ell}^\star\in\{0,1\}$ follows, $\forall k\in\mathcal K, \ell\in\mathcal L$, which verifies the equivalence between problems \eqref{P1_relaxed_eqv} and \eqref{P1_relaxed_eqv2}. It is worth mentioning that since setting $\rho$ significantly large at the very beginning may render this approach ineffective \cite{2000_Runarsson_penalty}, we initialize $\rho$ to a small value to find a good starting point and then solve problem \eqref{P1_relaxed_eqv2} iteratively with $\rho$ increasing with the iterations until $h\big({\{a_{k,\ell}^\star\}}\big) \rightarrow 0$. 

For any given $\rho$, problem \eqref{P1_relaxed_eqv2} is still hard to solve directly due to the non-concave objective function and the non-convex constraints in  \eqref{big_M_cons:d}, \eqref{P1_cons:b_eqv3}, and \eqref{P1_cons:c_eqv3}. Nevertheless, it is observed that either given or only optimizing $\{\mathbf v_\ell\}$, the resulting problem is more tractable. This motivates us to apply the BCD method as in the previous section to solve problem \eqref{P1_relaxed_eqv2} suboptimally by alternately optimizing $\tilde{\mathcal Z} \triangleq \mathcal Z\backslash\{\mathbf v_\ell\}$ and $\{\mathbf v_{\ell}\}$, elaborated as follows. 

\subsection{Optimizing $\tilde{\mathcal Z}$ for Given $\{\mathbf v_{\ell}\}$} 
With given $\{\mathbf v_\ell\}$, all the other variables in $\mathcal Z$ can be jointly optimized by solving the subproblem of \eqref{P1_relaxed_eqv2}, which is expressed as
\begin{subequations}\label{P1_relaxed_eqv2_sub1}
	\begin{eqnarray}
	&\underset{\eta,\tilde{\mathcal Z}}{\max}& \eta - \rho h\big(\{a_{k,\ell}\}\big)\\
	&\text{s.t.}& \eqref{P1_cons:e}, \eqref{P1_cons:g}, \eqref{big_M_cons}, \eqref{P1_cons:c_eqv3}, \eqref{P1_cons:d_eqv3}, \eqref{P1_cons:f_1}, \\
	&& \sum_{\ell\in\mathcal L}\left( f_{k,\ell} - g_{k,\ell}\right)  \geq \eta, \ \forall k\in\mathcal K,  \label{P1_relaxed_eqv2_sub1_cons:c} 
	\end{eqnarray}
\end{subequations}
where constraint \eqref{P1_relaxed_eqv2_sub1_cons:c} is the equivalent form of constraint \eqref{P1_cons:b_eqv3}, with the expressions of the concave functions $f_{k,\ell}$ and $g_{k,\ell}$ given by
\begin{align}
& f_{k,\ell} = \tau_\ell\log_2\left( \frac{\sum_{i\in \mathcal K}{\rm tr}\left(\mathbf X_{k,\ell}\tilde{\mathbf S}_{i,\ell}\right)}{\tau_\ell} + \frac{{\rm tr}\left(\mathbf X_{k,\ell}\mathbf S_{\mathrm E,\ell}\right)}{\tau_\ell} + \sigma_k^2 \right), \ \forall k\in\mathcal K, \ell\in\mathcal L, \label{express_f}\\ 
& g_{k,\ell} = \tau_\ell\log_2\left( \frac{\sum_{i\in \mathcal K\backslash\{k\}}{\rm tr}\left(\mathbf X_{k,\ell}\tilde{\mathbf S}_{i,\ell}\right)}{\tau_\ell} + \frac{{\rm tr}\left(\mathbf X_{k,\ell}\mathbf S_{\mathrm E,\ell}\right)}{\tau_\ell} + \sigma_k^2 \right), \ \forall k\in\mathcal K, \ell\in\mathcal L, \label{express_g}
\end{align}
respectively. We observe that the convex term $a_{k,\ell}^2$ in $h\big(\{a_{k,\ell}\}\big)$ makes the objective function non-concave while the concave term $g_{k,\ell}$ in constraint \eqref{P1_relaxed_eqv2_sub1_cons:c} makes this constraint non-convex. These, together with the rank constraints in \eqref{big_M_cons:d}, lead to the non-convexity of problem \eqref{P1_relaxed_eqv2_sub1}. To handle this problem, we leverage the SCA technique as in the previous section. Specifically, since the first-order Taylor expansion of any convex (concave) function at any point is its global lower (upper) bound, the following inequalities hold:
\begin{align}
& a_{k,\ell}^2 \geq -\left( a_{k,\ell}^r\right)^2 + 2a_{k,\ell}^ra_{k,\ell} \triangleq \chi^{\rm lb, \it r}\left(a_{k,\ell}\right), \ \forall k\in\mathcal K, \ell\in\mathcal L, \label{express_a_lower}\\ 
& g_{k,\ell}\left(\hat{\mathbf S}_{k,\ell}, \mathbf S_{\mathrm E, \ell}, \tau_\ell\right)  \leq \tau_\ell^r\log_2\left(\Upsilon_{k,\ell}^r\right) + \frac{\sum_{i\in \mathcal K\backslash\{k\}}{\rm tr}\left(\mathbf X_{k,\ell}\left( \tilde{\mathbf S}_{i,\ell} \!-\! \tilde{\mathbf S}_{i,\ell}^r\right)\right) + {\rm tr}\left(\mathbf X_{k,\ell}\left( \mathbf S_{\mathrm E,\ell} \!-\! \mathbf S_{\mathrm E,\ell}^r\right) \right)}{\Upsilon_{k,\ell}^r\ln2} \nonumber\\
& \hspace{1.5cm} + \left(\log_2\left(\Upsilon_{k,\ell}^r\right) - \frac{\Upsilon_{k,\ell}^r \!-\! \sigma_k^2 }{\Upsilon_{k,\ell}^r\ln2}\right)\left(\tau_\ell \!-\! \tau_\ell^r\right) \triangleq g_{k,\ell}^{\rm ub, \it r}\left(\hat{\mathbf S}_{k,\ell}, \mathbf S_{\mathrm E, \ell}, \tau_\ell\right), \ \forall k\in\mathcal K, \ell\in\mathcal L, \label{express_g_upper}
\end{align}
where $\hat{\mathbf S}_{k,\ell}$ denotes the collection of the variables $\left\lbrace \tilde{\mathbf S}_{i,\ell}\right\rbrace_{\forall i\in\mathcal K\backslash\{k\}}$ and $\Upsilon_{k,\ell}^r = \frac{\sum_{i\in \mathcal K\backslash\{k\}}{\rm tr}\left(\mathbf X_{k,\ell}\tilde{\mathbf S}_{i,\ell}^r\right)}{\tau_\ell^r} + \frac{{\rm tr}\left(\mathbf X_{k,\ell}\mathbf S_{\mathrm E,\ell}^r\right)}{\tau_\ell^r} + \sigma_k^2$. In addition, $a_{k,\ell}^r$, $\tilde{\mathbf S}_{i,\ell}^r$, $\mathbf S_{\mathrm E,\ell}^r$, and $\tau_\ell^r$ represent the given local points in the $r$-th iteration of SCA. 

By replacing the term $a_{k,\ell}^2$ in $h\big(\{a_{k,\ell}\}\big)$ with $\chi^{\rm lb, \it r}\left(a_{k,\ell}\right)$ and the term $g_{k,\ell}$ in \eqref{P1_relaxed_eqv2_sub1_cons:c} with $g_{k,\ell}^{\rm ub, \it r}\left(\hat{\mathbf S}_\ell, \mathbf S_{\mathrm E, \ell}, \tau_\ell\right)$, a performance lower bound of problem \eqref{P1_relaxed_eqv2_sub1} can be obtained by solving 
\vspace{-1mm}
\begin{subequations}\label{P1_relaxed_eqv2_sub1_sca}
	\setlength\abovedisplayskip{-3pt}
	\begin{eqnarray}
	&\underset{\eta,\tilde{\mathcal Z}}{\max}&  \eta - \rho h^{\rm ub, \it r}\big(\{a_{k,\ell}\}\big) \label{P1_relaxed_eqv2_sub1_sca_obj}\\
	&\text{s.t.}& \eqref{P1_cons:e}, \eqref{P1_cons:g}, \eqref{big_M_cons}, \eqref{P1_cons:c_eqv3}, \eqref{P1_cons:d_eqv3}, \eqref{P1_cons:f_1}, \\
	&& \sum_{\ell\in\mathcal L}\left( f_{k,\ell} - g_{k,\ell}^{\rm ub, \it r}\left(\hat{\mathbf S}_\ell, \mathbf S_{\mathrm E, \ell}, \tau_\ell\right)\right) \geq \eta, \ \forall k\in\mathcal K,  \label{P1_relaxed_eqv2_sub1_cons:c_sca}
	\end{eqnarray}
\end{subequations}
where $h^{\rm ub, \it r}\big(\{a_{k,\ell}\}\big) \triangleq \sum_{\ell \in\mathcal L}\sum_{k \in\mathcal K} \left(a_{k,\ell} - \chi^{\rm lb, \it r}\left(a_{k,\ell}\right)\right)$. If we drop the non-convex rank constraints in \eqref{big_M_cons:d}, problem \eqref{P1_relaxed_eqv2_sub1_sca} is reduced to a convex SDP that can be solved exactly using off-the-shelf solvers, e.g., CVX \cite{2004_S.Boyd_cvx}. However, the obtained $\left\lbrace\tilde{\mathbf S}_{k,\ell}\right\rbrace$ cannot be guaranteed to satisfy constraint \eqref{big_M_cons:d}. Therefore, instead of dropping constraint \eqref{big_M_cons:d}, we equivalently transform it into \looseness=-1 
\begin{align}\label{big_M_cons:d_eqv}
{\rm tr}\left(\tilde{\mathbf S}_{k,\ell}\right) - \left\|\tilde{\mathbf S}_{k,\ell}\right\|_2 \leq 0, \ \forall k\in\mathcal K, \ell\in\mathcal L,
\end{align}
which is a reverse convex constraint. Similar to problem \eqref{P1_relaxed_eqv2}, we incorporate constraint \eqref{big_M_cons:d_eqv} into the objective function in \eqref{P1_relaxed_eqv2_sub1_sca_obj} by introducing a penalty parameter $\mu > 0$ and then convert problem \eqref{P1_relaxed_eqv2_sub1_sca} to  
\begin{subequations}\label{P1_relaxed_eqv2_sub1_sca_eqv}
	\begin{eqnarray}
	&\underset{\eta,\tilde{\mathcal Z}}{\max}&  \eta - \rho h^{\rm ub, \it r}\big(\{a_{k,\ell}\}\big) - \mu q\left(\left\lbrace\tilde{\mathbf S}_{k,\ell}\right\rbrace\right)  \\
	&\text{s.t.}& \eqref{P1_cons:e}, \eqref{P1_cons:g},  \eqref{big_M_cons:a}-\eqref{big_M_cons:c}, \eqref{P1_cons:c_eqv3}, \eqref{P1_cons:d_eqv3}, \eqref{P1_cons:f_1}, \eqref{P1_relaxed_eqv2_sub1_cons:c_sca},
	\end{eqnarray}
\end{subequations}
where $q\left(\left\lbrace\tilde{\mathbf S}_{k,\ell}\right\rbrace\right) \triangleq \sum_{\ell\in\mathcal L}\sum_{k \in\mathcal K}\left( {\rm tr}\left(\tilde{\mathbf S}_{k,\ell}\right) - \left\|\tilde{\mathbf S}_{k,\ell}\right\|_2\right)$. When $\mu\rightarrow\infty$, solving problem \eqref{P1_relaxed_eqv2_sub1_sca_eqv} yields an identical solution to problem \eqref{P1_relaxed_eqv2_sub1_sca}. Despite having a convex feasible set, problem \eqref{P1_relaxed_eqv2_sub1_sca_eqv} is non-convex due to the convexity of the term $\left\|\tilde{\mathbf S}_{k,\ell}\right\|_2$ in $q\left(\left\lbrace\tilde{\mathbf S}_{k,\ell}\right\rbrace\right)$. By replacing $\left\|\tilde{\mathbf S}_{k,\ell}\right\|_2$ with its first-order Taylor expansion-based
lower bound, we can approximate problem \eqref{P1_relaxed_eqv2_sub1_sca_eqv} as \begin{subequations}\label{P1_relaxed_eqv2_sub1_sca_eqv_sca}
	\begin{eqnarray}
	&\underset{\eta,\tilde{\mathcal Z}}{\max}&  \eta - \rho h^{\rm ub, \it r}\big(\{a_{k,\ell}\}\big) - \mu q^{\rm ub, \it r}\left(\left\lbrace\tilde{\mathbf S}_{k,\ell}\right\rbrace\right) \\
	&\text{s.t.}& \eqref{P1_cons:e}, \eqref{P1_cons:g},  \eqref{big_M_cons:a}-\eqref{big_M_cons:c}, \eqref{P1_cons:c_eqv3}, \eqref{P1_cons:d_eqv3}, \eqref{P1_cons:f_1}, \eqref{P1_relaxed_eqv2_sub1_cons:c_sca},
	\end{eqnarray}
\end{subequations}
where $q^{\rm ub, \it r}\left(\left\lbrace\tilde{\mathbf S}_{k,\ell}\right\rbrace\right) \triangleq \sum_{\ell\in\mathcal L}\sum_{k \in\mathcal K}\left( {\rm tr}\left(\tilde{\mathbf S}_{k,\ell}\right) \!-\! \left\|\tilde{\mathbf S}_{k,\ell}^r\right\|_2 \!-\! \left(\tilde{\mathbf s}_{k,\ell}^{\max, r}\right)^H\left(\tilde{\mathbf S}_{k,\ell} \!-\! \tilde{\mathbf S}_{k,\ell}^r\right)\tilde{\mathbf s}_{k,\ell}^{\max, r}\right)$ with $\tilde{\mathbf s}_{k,\ell}^{\max, r}$ being the eigenvector that corresponds to the largest
eigenvalue of $\tilde{\mathbf S}_{k,\ell}^r$. Since problem \eqref{P1_relaxed_eqv2_sub1_sca_eqv_sca} is a convex SDP, standard solvers such as CVX \cite{2004_S.Boyd_cvx} can be used to find its optimal solution. 

\begin{algorithm}[!t]  
	\caption{Proposed algorithm for problem \eqref{P1_relaxed_eqv2_sub1}}  \label{Alg1}  
	\begin{algorithmic}[1]
		\STATE Initialize $\tilde{\mathcal Z}^0 \triangleq \left\lbrace \left\lbrace \tilde{\mathbf S}_{k,\ell}^0\right\rbrace, \{\mathbf S_{k,\ell}^0\}, \{\mathbf S_{\mathrm E,\ell}^0\}, \{ a_{k,\ell}^0\}, \{\tau_\ell^0\}\right\rbrace$, $\mu > 0$, and $c_1 > 1$.  
		\REPEAT
		\STATE Set $r = 0$. 
		\REPEAT 
		\STATE Obtain $\tilde{\mathcal Z}^{r+1}$ by solving problem \eqref{P1_relaxed_eqv2_sub1_sca_eqv_sca} with given $\tilde{\mathcal Z}^{r}$. 
		\STATE $r \leftarrow r + 1$.
		\UNTIL The fractional increase of the objective value of problem \eqref{P1_relaxed_eqv2_sub1_sca_eqv_sca} between two consecutive iterations falls below a threshold $\epsilon_1 > 0$. 
		\STATE $\tilde{\mathcal Z}^0 \leftarrow \tilde{\mathcal Z}^r$ and $\mu\leftarrow c_1\mu$.
		\UNTIL $q\left(\left\lbrace\tilde{\mathbf S}_{k,\ell}^r\right\rbrace\right)$ is below a threshold $\varsigma_1 > 0$. 
		\STATE Output $\tilde{\mathcal Z}^r$ as a locally optimal solution of problem \eqref{P1_relaxed_eqv2_sub1}. 
	\end{algorithmic} 
\end{algorithm}

Based on the above, we provide in Algorithm \ref{Alg1} the details of solving problem suboptimally \eqref{P1_relaxed_eqv2_sub1} via combining the SCA and the penalty method, where $c_1 > 1$ is a scaling factor. The inner loop of Algorithm \ref{Alg1} is used to iteratively solve problem \eqref{P1_relaxed_eqv2_sub1_sca_eqv_sca} under fixed $\mu$, whose convergence is guaranteed since the objective value is non-decreasing over the iterations and also bounded from above. In the outer loop, by iteratively increasing $\mu$ via $\mu\leftarrow c_1\mu$, we enforce $q\left(\left\lbrace\tilde{\mathbf S}_{k,\ell}\right\rbrace\right) \rightarrow 0$, such that the obtained solution satisfy the rank constraints on $\left\lbrace\tilde{\mathbf S}_{k,\ell}\right\rbrace$. In this way, Algorithm \ref{Alg1} is guaranteed to converge to a stationary point of problem \eqref{P1_relaxed_eqv2_sub1} \cite{2016_Sun_stationary}. 
\vspace{-3mm}
\begin{rem}\label{rem1}
	It is worth mentioning that if there are no phase errors at the IRS elements, the matrix $\mathbf X_{k,\ell}$ in constraint \eqref{P1_relaxed_eqv2_sub1_cons:c_sca} can be replaced by $\mathbf H_k^H\mathbf v_\ell\mathbf v_\ell^H\mathbf H_k \triangleq \hat{\mathbf X}_{k,\ell}$ and we have ${\rm rank}\left(\hat{\mathbf X}_{k,\ell}\right) = 1$, $\forall k\in\mathcal K, \ell\in\mathcal L$. With this condition, it can be proved that, for arbitrary direct and cascaded channels, if the optimal solution obtained by solving problem \eqref{P1_relaxed_eqv2_sub1_sca} with the rank constraint \eqref{big_M_cons:d} removed (or equivalently, problem  \eqref{P1_relaxed_eqv2_sub1_sca_eqv_sca} with $\mu = 0$) violates constraint \eqref{big_M_cons:d}, we can always construct an alternative optimal solution that satisfies constraint \eqref{big_M_cons:d} by using $\mathbf S_{\mathrm E, \ell}$ to absorb the non-rank-one part of each $\tilde{\mathbf S}_{k,\ell}$. The corresponding proof is similar to that in \cite[Appendix B]{2020_Xianghao_secure_rank1}, and we omit it for brevity. However, in the presence of the phase errors, we cannot prove the above result by following the same derivation as in \cite[Appendix B]{2020_Xianghao_secure_rank1} since $\mathbf X_{k,\ell}$ is generally of high rank. Despite this, almost all of our simulations show that solving problem \eqref{P1_relaxed_eqv2_sub1_sca_eqv_sca} with even a sufficiently small $\mu$ via CVX can yield a rank-one optimal solution. Thus, the characterization of the optimal solution structure of problem \eqref{P1_relaxed_eqv2_sub1_sca_eqv_sca} (or problem \eqref{P1_relaxed_eqv2_sub1_sca})  deserves further study. 
\end{rem}
\vspace{-3mm}    

\vspace{-1mm}
\subsection{Optimizing $\{\mathbf v_{\ell}\}$ for Given $\tilde{\mathcal Z}$} Given any feasible $\tilde{\mathcal Z}$, by introducing slack variables $\{\lambda_{k,\ell}\}$ and ignoring the constant term $-\rho h\left(\{a_{k,\ell}\}\right)$ in the objective function of problem \eqref{P1_relaxed_eqv2}, we can equivalently express the subproblem of \eqref{P1_relaxed_eqv2} w.r.t. $\{\mathbf v_{\ell}\}$ as   
\begin{subequations}\label{P1_relaxed_eqv2_sub2}
	\setlength\abovedisplayskip{5pt}
	\setlength\belowdisplayskip{5pt}
	\begin{eqnarray}
	&\hspace{-3mm}\underset{\eta,\{\mathbf v_\ell\},\{\lambda_{k,\ell}\}}{\max}& \eta \\
	&\text{s.t.}& \hspace{-6mm} \eqref{P1_cons:h_relaxed}, \eqref{P1_cons:c_eqv3}, \\
	&& \hspace{-6mm} \sum_{\ell\in\mathcal L'}\tau_\ell\log_2\left( 1 + \lambda_{k,\ell} \right) \geq \eta, \ \forall k\in\mathcal K,  \label{P1_relaxed_eqv2_sub2_cons:c}\\
	&& \hspace{-6mm} \frac{\frac{{\rm tr}\left(\mathbf X_{k,\ell}\tilde{\mathbf S}_{k,\ell}\right)}{\tau_\ell}}{\lambda_{k,\ell}} \geq \frac{\sum_{i\in \mathcal K\backslash\{k\}}{\rm tr}\left(\mathbf X_{k,\ell}\tilde{\mathbf S}_{i,\ell}\right)}{\tau_\ell} + \frac{{\rm tr}\left(\mathbf X_{k,\ell}\mathbf S_{\mathrm E,\ell}\right)}{\tau_\ell} + \sigma_k^2, \ \forall k\in\mathcal K, \ell\in\mathcal L', \label{P1_relaxed_eqv2_sub2_cons:d}
	\end{eqnarray}
\end{subequations}
where $\mathcal L' = \{\ell|\tau_\ell > 0\} \subseteq \mathcal L$, and the constraints in \eqref{P1_relaxed_eqv2_sub2_cons:c} and \eqref{P1_relaxed_eqv2_sub2_cons:d} are transformed from those in \eqref{P1_cons:b_eqv3}, which incurs no loss of optimality since there always exists an optimal solution to problem \eqref{P1_relaxed_eqv2_sub2} that makes the constraints in \eqref{P1_relaxed_eqv2_sub2_cons:d} satisfied with equality. Observe that the optimization variables $\{\mathbf v_\ell\}$ are not exposed in the current forms of constraints \eqref{P1_cons:c_eqv3} and \eqref{P1_relaxed_eqv2_sub2_cons:d}. To facilitate the solution development of problem \eqref{P1_relaxed_eqv2_sub2}, we recast \eqref{P1_cons:c_eqv3} and \eqref{P1_relaxed_eqv2_sub2_cons:d} as  
\begin{align}
& \sum_{\ell\in\mathcal L'}\left(\sum_{k\in\mathcal K}\mathbf v_\ell^H\mathbf A_{j,k,\ell}\mathbf v_\ell + \mathbf v_\ell^H\mathbf B_{j,\mathrm E,\ell}\mathbf v_\ell\right) \geq E, \ \forall j\in\mathcal J, \label{P1_cons:c_eqv3_eqv}\\
& \frac{\frac{\mathbf v_\ell^H\mathbf C_{k,k,\ell}\mathbf v_\ell}{\tau_\ell}}{\lambda_{k,\ell}} \geq \frac{\sum_{i\in \mathcal K\backslash\{k\}}\mathbf v_\ell^H\mathbf C_{k,i,\ell} \mathbf v_\ell}{\tau_\ell} + \frac{\mathbf v_\ell^H\mathbf D_{k,\mathrm E,\ell}\mathbf v_\ell}{\tau_\ell} + \sigma_k^2, \ \forall k\in\mathcal K, \ell\in\mathcal L', \label{P1_relaxed_eqv2_sub2_cons:d_eqv}
\end{align}%
where $\mathbf A_{j,k,\ell} = {\rm diag}\left(\mathbf G_j\tilde{\mathbf s}_{k,\ell}\right)\mathbf Z\left({\rm diag}\left(\mathbf G_j\tilde{\mathbf s}_{k,\ell}\right)\right)^H$ if $\tilde{\mathbf S}_{k,\ell} \neq \mathbf 0$ and $\mathbf A_{j,k,\ell} = \mathbf 0$ otherwise, $\mathbf B_{j,\mathrm E,\ell} = \sum_{m=1}^{\pi_{\rm E,\ell}}b_{\ell,m}{\rm diag}\left(\mathbf G_j\mathbf s_{\rm E,\ell,m}\right)\mathbf Z\left( {\rm diag}\left(\mathbf G_j\mathbf s_{\rm E,\ell,m}\right)\right)^H$ if ${\mathbf S}_{\mathrm E,\ell} \neq \mathbf 0$ and $\mathbf B_{j,\mathrm E,\ell} = \mathbf 0$ otherwise, $\mathbf C_{k,i,\ell} = {\rm diag}\left(\mathbf H_k\tilde{\mathbf s}_{i,\ell}\right)\mathbf Z\left({\rm diag}\left(\mathbf H_k\tilde{\mathbf s}_{i,\ell}\right)\right)^H$ if $\tilde{\mathbf S}_{i,\ell} \neq \mathbf 0$ and $\mathbf C_{k,i,\ell} = \mathbf 0$ otherwise, and $\mathbf D_{k,\mathrm E,\ell} = \sum_{m=1}^{\pi_{\rm E,\ell}}b_{\ell,m}{\rm diag}\left(\mathbf H_k\mathbf s_{\rm E,\ell,m}\right)\mathbf Z\left( {\rm diag}\left(\mathbf H_k\mathbf s_{\rm E,\ell,m}\right)\right)^H$ if $\mathbf S_{\mathrm E,\ell} \neq \mathbf 0$ and $\mathbf D_{k,\mathrm E,\ell} = \mathbf 0$ otherwise, $\forall j\in\mathcal J, k,i\in\mathcal K, \ell\in\mathcal L$. In addition, $\tilde{\mathbf s}_{k,\ell}$ is obtained from $\tilde{\mathbf S}_{k,\ell}$ by performing the Cholesky
decomposition, i.e., $\tilde{\mathbf S}_{k,\ell} = \tilde{\mathbf s}_{k,\ell}\tilde{\mathbf s}_{k,\ell}^H$, and $\{b_{\ell,m}\}$ and $\{\mathbf s_{\mathrm E,\ell,m}\}$ are obtained from the eigenvalue
decomposition of $\mathbf S_{\mathrm E,\ell}$ with $\mathbf S_{\rm E,\ell} = \sum_{m=1}^{\pi_{\rm E,\ell}}b_{\ell,m}\mathbf s_{\rm E,\ell,m}\mathbf s_{\rm E,\ell,m}^H$ and $\pi_{\mathrm E,\ell} = {\rm rank}\left(\mathbf S_{\mathrm E,\ell}\right)$. The proofs of the equivalence between \eqref{P1_cons:c_eqv3} and \eqref{P1_cons:c_eqv3_eqv} and between \eqref{P1_relaxed_eqv2_sub2_cons:d} and \eqref{P1_relaxed_eqv2_sub2_cons:d_eqv} are similar to that in Appendix \ref{Appen_B} for Lemma \ref{lem} and are omitted here for brevity.

It is obvious that constraints \eqref{P1_cons:c_eqv3_eqv} and \eqref{P1_relaxed_eqv2_sub2_cons:d_eqv} are non-convex since the quadratic terms in the left-hand-sides of them are convex w.r.t. $\mathbf v_\ell$, which motivates us to convexify these two constraints via the SCA technique. To be specific, by replacing the left-hand-sides of the non-convex constraints \eqref{P1_cons:c_eqv3_eqv} and \eqref{P1_relaxed_eqv2_sub2_cons:d_eqv} with their respective first-order Taylor expansions at the given local points $\{\mathbf v_\ell^q\}$ in the $q$-th iteration of SCA, \eqref{P1_cons:c_eqv3_eqv} and \eqref{P1_relaxed_eqv2_sub2_cons:d_eqv} can be approximated as the following convex constraints: 
\begin{align}
& \sum_{\ell\in\mathcal L'}\left(\sum_{k\in\mathcal K} \mathcal F^{\rm lb,\it q}_{\mathbf A_{j,k,\ell}}\left(\mathbf v_\ell\right) + \mathcal F^{\rm lb,\it q}_{\mathbf B_{j,\mathrm E,\ell}}\left(\mathbf v_\ell\right)\right) \geq E, \ \forall j\in\mathcal J, \label{P1_cons:c_eqv3_eqv_sca}\\
& \mathcal G^{\rm lb, \it q}\left(\mathbf v_\ell, \lambda_{k,\ell}\right) \geq \frac{\sum_{i\in \mathcal K\backslash\{k\}}\mathbf v_\ell^H\mathbf C_{k,i,\ell} \mathbf v_\ell}{\tau_\ell} + \frac{\mathbf v_\ell^H\mathbf D_{k,\mathrm E,\ell}\mathbf v_\ell}{\tau_\ell} + \sigma_k^2, \ \forall k\in\mathcal K, \ell\in\mathcal L', \label{P1_relaxed_eqv2_sub2_cons:d_eqv_sca}
\end{align}%
where $\mathcal F^{\rm lb,\it q}_{\mathbf R}\left(\mathbf v_\ell\right) \triangleq 2{\rm Re}\left\lbrace \mathbf v_\ell^H\mathbf R\mathbf v_\ell^q\right\rbrace - \big(\mathbf v_\ell^q\big)^H\mathbf R\mathbf v_\ell^q$, $\mathbf R \in \left\lbrace\mathbf A_{j,k,\ell}, \mathbf B_{j,\mathrm E,\ell} \right\rbrace$, and $\mathcal G^{\rm lb, \it q}\left(\mathbf v_\ell, \lambda_{k,\ell}\right) \triangleq \frac{2{\rm Re}\left\lbrace\mathbf v_\ell^H\mathbf C_{k,k,\ell}\mathbf v_\ell^q\right\rbrace}{\tau_\ell\lambda_{k,\ell}^q} - \frac{\big(\mathbf v_\ell^q\big)^H\mathbf C_{k,k,\ell}\mathbf v_\ell^q}{\tau_\ell\left( \lambda_{k,\ell}^q\right)^2}\lambda_{k,\ell}$. Then, a locally optimal solution of problem \eqref{P1_relaxed_eqv2_sub2} can be obtained by iteratively solving the following convex QCQP via readily available solvers (e.g., CVX \cite{2004_S.Boyd_cvx}) until convergence is declared \cite{2010_Dinh_SCA_converge}.
\begin{align}\label{P1_relaxed_eqv2_sub2_sca}
\underset{\eta,\{\mathbf v_\ell\},\{\lambda_{k,\ell}\}}{\max} \hspace{3mm} \eta \hspace{8mm}
\text{s.t.} \hspace{3mm}  \eqref{P1_cons:h_relaxed}, \eqref{P1_relaxed_eqv2_sub2_cons:c}, \eqref{P1_cons:c_eqv3_eqv_sca}, \eqref{P1_relaxed_eqv2_sub2_cons:d_eqv_sca}. 
\end{align}

\begin{algorithm}[!t]  
	\caption{Proposed algorithm for (P1)}  \label{Alg2}  
	\begin{algorithmic}[1]
		\STATE Initialize $\tilde{\mathcal Z}^0$, $\{\mathbf v_\ell^0\}$, $\rho > 0$, and $c_2 > 1$.  
		\REPEAT
		\STATE Set $i = 0$. 
		\REPEAT 
		\STATE Solve problem \eqref{P1_relaxed_eqv2_sub1} via Algorithm \ref{Alg1} for given $\tilde{\mathcal Z}^i$ and $\{\mathbf v_\ell^i\}\}$, and denote the obtained locally optimal solution as $\tilde{\mathcal Z}^{i+1}$. \label{step_BCD_1}
		\STATE Solve problem \eqref{P1_relaxed_eqv2_sub2} via SCA for given $\tilde{\mathcal Z}^{i+1}$ and $\{\mathbf v_\ell^i\}\}$, and denote the obtained locally optimal solution as $\{\mathbf v_\ell^{i+1}\}$. \label{step_BCD_2}
		\STATE $i \leftarrow i + 1$.
		\UNTIL The fractional increase of the objective value of problem \eqref{P1_relaxed_eqv2} is smaller than a threshold $\epsilon_2 > 0$. 
		\STATE $\tilde{\mathcal Z}^0 \leftarrow \tilde{\mathcal Z}^i$, $\{\mathbf v_\ell^0\} \leftarrow \{\mathbf v_\ell^i\}$, and $\rho\leftarrow c_2\rho$.
		\UNTIL $h\big(\{a_{k,\ell}^i\}\big)$ is below a threshold $\varsigma_2 > 0$. 
		\STATE Set $\left[ \hat{\mathbf v}_\ell\right]_n = \left[\mathbf v_\ell^i\right]_n/\left|\left[ \mathbf v_\ell^i\right]_n\right|$, set $\hat{\mathbf W}_{\mathrm E,\ell} = \tilde {\mathbf S}_{\mathrm E,\ell}^i/\tau_\ell^i$ if $\tau_\ell^i > 0$ and $\hat{\mathbf W}_{\mathrm E,\ell} = \mathbf 0$ otherwise, set $\hat{\mathbf W}_{k,\ell} = \tilde {\mathbf S}_{k,\ell}^i/\tau_\ell^i$ if $\tau_\ell^i > 0$ and $\hat{\mathbf W}_{k,\ell} = \mathbf 0$ otherwise, and decompose $\hat{\mathbf W}_{k,\ell}$ as $\hat{\mathbf W}_{k,\ell} = \hat{\mathbf w}_{k,\ell}\hat{\mathbf w}_{k,\ell}^H$ via the Cholesky
		decomposition, $\forall \ell\in\mathcal L, n\in\mathcal N, k\in\mathcal K$. \label{step_constr_1}
		\STATE Compute $\hat\eta$ based on $\hat{\mathcal Z} \triangleq \left\lbrace \{ \hat{\mathbf w}_{k,\ell}\}, \{\hat{\mathbf W}_{\mathrm E,\ell}\},\{ a_{k,\ell}^i\}, \{\tau_\ell^i\},\{\hat{\mathbf v}_\ell\}\right\rbrace$, and output $\left\lbrace\hat\eta, \hat{\mathcal Z}\right\rbrace $ as a suboptimal solution of problem (P1). \label{step_constr_2} 
	\end{algorithmic} 
\end{algorithm}

\subsection{Overall Algorithm}
Based on the above results, we summarize the details of our proposed algorithm for (P1) in Algorithm \ref{Alg2}. For any given $\rho$, the BCD inner loop of Algorithm \ref{Alg2} solves problem \eqref{P1_relaxed_eqv2} by alternately solving problems \eqref{P1_relaxed_eqv2_sub1} and \eqref{P1_relaxed_eqv2_sub2} and is guaranteed to converge to a stationary point of problem \eqref{P1_relaxed_eqv2} \cite{2001_Tseng_BCD}. In the outer loop, we gradually increase $\rho$ to a sufficiently large value via $\rho \leftarrow c_2\rho$ to make $h\big(\{a_{k,\ell}\}\big) \rightarrow 0$, thereby ensuring $a_{k,\ell} \in \{0,1\}$, $\forall k\in\mathcal K, \ell\in\mathcal L$. As a consequence, after the convergence of the outer loop, we can obtain a stationary solution of problem \eqref{P1_relaxed_eqv2} satisfying the binary constraints on $\{a_{k,\ell}\}$. Since the obtained $\{\mathbf v_\ell^i\}$ may not satisfy the unit-modulus constraints of (P1), we set $\left[ \hat{\mathbf v}_\ell\right]_n = \left[\mathbf v_\ell^i\right]_n/\left|\left[ \mathbf v_\ell^i\right]_n\right|$, $\forall \ell\in\mathcal L, n\in\mathcal N$, without violating any other constraints of (P1). Then, by performing the remaining operations in steps \ref{step_constr_1} and \ref{step_constr_2} of Algorithm \ref{Alg2}, we can obtain a suboptimal solution of (P1). 

The computational complexity of Algorithm \ref{Alg2} is analyzed as follows. In each inner loop iteration, the main complexity lies in steps \ref{step_BCD_1} and \ref{step_BCD_2}. The computational cost of step \ref{step_BCD_1} for solving problem \eqref{P1_relaxed_eqv2_sub1} via Algorithm \ref{Alg1} is $\mathcal O\left(I_{\rm out}^1I_{\rm inn}^1\sqrt{M}\log_2\left({1}/{\varepsilon}\right)\big(\omega M^3 + \omega^2 M^2 + \omega^3\big)\right)$ \cite{2010_Imre_SDR_complexity}, where $I_{\rm inn}^1$ and $I_{\rm out}^1$ denote the numbers of inner and outer iterations required for the convergence of Algorithm \ref{Alg1}, respectively, $\varepsilon$ is the solution accuracy, and $\omega \triangleq 3KL + K + L + J$. The complexity of step \ref{step_BCD_2} for solving problem \eqref{P1_relaxed_eqv2_sub2} via SCA is $\mathcal O\Big(I_{\rm s}\sqrt{JN + JL + KL}\log_2\left({1}/{\varepsilon}\right)\big(JKL^4N^3 + JL^4N^4 + K^2L^4N^2 + K^3L^3\big)\Big)$ \cite{2014_K.wang_complexity}, where $I_{\rm s}$ stands for the number of iterations required for the convergence of SCA. Therefore, the overall complexity of Algorithm \ref{Alg2} is about $\mathcal O\Big[I_{\rm out}^2I_{\rm inn}^2\log_2\left({1}/{\varepsilon}\right)\Big(I_{\rm out}^1I_{\rm inn}^1\sqrt{M}\big(\omega M^3 + \omega^2 M^2 + \omega^3\big) + I_{\rm s}\sqrt{JN + JL + KL}\big(JKL^4N^3 + JL^4N^4 + K^2L^4N^2 + K^3L^3\big)\Big)\Big]$, with $I_{\rm inn}^2$ and $I_{\rm out}^2$ denoting the numbers of inner and outer iterations required for the convergence of Algorithm \ref{Alg2}, respectively. 

\section{Proposed Algorithm for (P2)}\label{Sec:P2_solution}
We note that (P2) differs from (P1) in the sense that it does not have constraints on $\sum_{\ell \in\mathcal L} a_{k,\ell}$, $\forall k\in\mathcal K$, as in constraint \eqref{P1_cons:g} of (P1), which enables us to simplify (P2) by removing the binary variables $\{a_{k,\ell}\}$. In other words, we have the following theorem. 
\begin{theo}\label{theo2}
	\vspace{-3mm}
	Problem (P2) shares the same optimal value with its simplified version, denoted by (P2') which is obtained by removing $\{a_{k,\ell}\}$ in (P2).
\end{theo}
\begin{proof}
	\vspace{-3mm}
	Denote by $\bar{\eta}$ and $\grave{\eta}$ the optimal values of (P2) and (P2'), respectively. First, we have $\bar{\eta} \geq \grave{\eta}$ since (P2') is actually a special case of (P2) with $a_{k,\ell} = 1$, $\forall k \in\mathcal K, \ell\in\mathcal L$. Next, denote $\left\lbrace \bar\eta, \{\bar{\mathbf w}_{k,\ell}\}, \{\bar{\mathbf W}_{\mathrm E,\ell}\},\{ \bar a_{k,\ell}\}, \{\bar{\tau}_\ell\},\{\bar{\mathbf v}_\ell\}\right\rbrace$ as an arbitrary optimal solution to (P2). Let $\breve{\mathbf w}_{k,\ell} = \bar{\mathbf w}_{k,\ell}$ if $\bar a_{k,\ell} = 1$ and $\breve{\mathbf w}_{k,\ell} = \mathbf 0$ otherwise, $\forall k\in\mathcal K, \ell\in\mathcal L$. It is easy to verify that $\left\lbrace\bar\eta, \{\breve{\mathbf w}_{k,\ell}\}, \{\bar{\mathbf W}_{\mathrm E,\ell}\}, \{\bar{\tau}_\ell\},\{\bar{\mathbf v}_\ell\}\right\rbrace$ is a feasible solution to (P2'). Then, it follows that $\bar\eta \leq \grave{\eta}$. This, together with $\bar\eta \geq \grave{\eta}$, yields $\bar\eta = \grave{\eta}$. Theorem \ref{theo2} is thus proved. 
	\vspace{-2mm} 
\end{proof}
Based on Theorem \ref{theo2}, we only need to focus on solving (P2'). Since (P2') is similar to but much simpler than (P1), Algorithm \ref{Alg2} for (P1) can be modified to solve (P2'). Furthermore, the computational complexity of solving (P2') is much lower than that of solving (P1) since (P2') does not involve binary variables $\{a_{k,\ell}\}$. The details are omitted due to the space limitation. Denote by $\acute{\mathcal Z} \triangleq \left\lbrace\acute\eta, \{\acute{\mathbf w}_{k,\ell}\}, \{\acute{\mathbf W}_{\mathrm E,\ell}\}, \{\acute{\tau}_\ell\},\{\acute{\mathbf v}_\ell\}\right\rbrace$ the obtained solution of (P2'). Let $\acute a_{k,\ell} = 1$ if $\acute{\mathbf w}_{k,\ell} \neq 0$ and $\acute a_{k,\ell} = 0$ otherwise, $\forall k\in\mathcal K, \ell\in\mathcal L$. By doing so, we obtain a suboptimal solution $\left\lbrace \acute{\mathcal Z}, \left\lbrace \acute a_{k,\ell} \right\rbrace \right\rbrace$ of (P2).   

\begin{figure}[!t]
	\centering
	\includegraphics[width=0.66\textwidth]{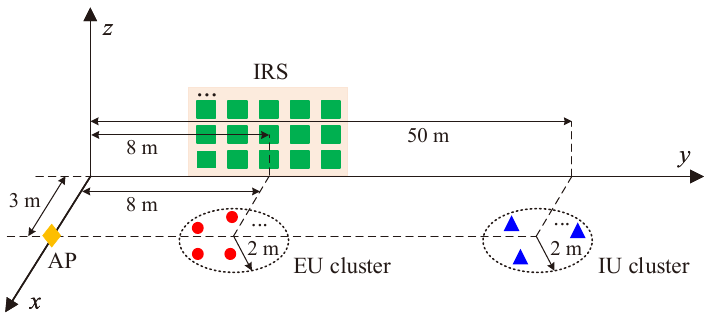}
	\caption{Simulation setup. The AP, IRS elements, EUs, and IUs are marked by orange `$\blacklozenge$', green '$\blacksquare$'s, red '$\bullet$'s, and blue '$\blacktriangle$'s, respectively.} \label{Fig:simulation_setup}
	\vspace{-2mm}
\end{figure}

\section{Simulation Results}\label{Sec:simulation}
In this section, simulations are presented to evaluate the performance of our proposed UG schemes. As illustrated in Fig. \ref{Fig:simulation_setup}, we consider a three-dimensional (3D) coordinate setup with the locations of the AP and the IRS being $\left(3, 0, 0\right)$ and $\left(0, 8, 0 \right)$ measured in meter (m), respectively. The EUs and the IUs are randomly and uniformly distributed in two different circular regions centered at $\left(3, 8, 0\right)$ m and $\left(3, 50, 0\right)$ m, respectively, with identical radii of $2$ m. Each channel response is assumed to comprise two types of radio fading: large-scale and small-scale. The large-scale fading is modeled as ${\rm PL}(d) = C_0/d^{\alpha}$ \cite{2019_Qingqing_Joint}, where $C_0$, $d$, and $\alpha$ denote the path loss at the reference distance of $1$ m, the link distance, and the path loss exponent, respectively. We set $C_0 = -30$ dB for all the links, $\alpha = 3.5$ for the direct links, and $\alpha = 2.2$ for the IRS-related links, respectively. Furthermore, the small-scale fading is characterized by Rayleigh fading for the direct links while Rician fading for the IRS-related links with a Rician factor of $3$ dB. Unless otherwise stated, other parameters are set as $\sigma_k^2 = -80$ dBm, $\forall k\in \mathcal K$, $M = 4$, $N = 40$, $P = 43$ dBm, $T = 1$ s, $\mu = \rho = 10^{-2}$, $c_1 = c_2 = 10$, $\epsilon_1 = \epsilon_2 = 10^{-4}$, and $\varsigma_1 = \varsigma_2 = 10^{-7}$. 

\begin{figure}[!t]
	\centering
	\includegraphics[scale=0.71]{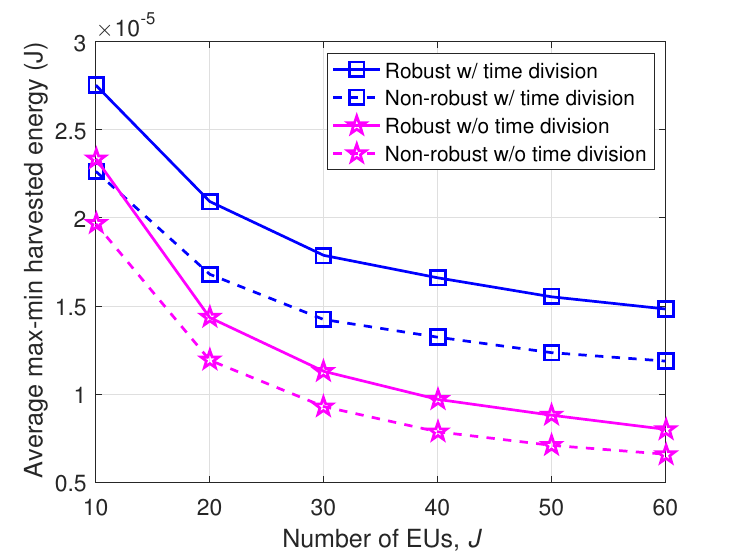}\vspace{-1.5mm}
	\caption{Average max-min harvested energy versus the number of EUs for $L = 5$.}
	\label{fig:E_vs_EUs}
	\vspace{-3mm}
\end{figure}

\vspace{-1mm}
\subsection{Achievable Max-Min Harvested Energy}
We first provide a numerical comparison of the max-min harvested energy achievable by the following schemes: 1) \textbf{Robust w/ time division}: the algorithm proposed in Section \ref{Sec:feasi_check} for problem \eqref{prob:feasi_E}; 2) \textbf{Non-robust w/ time division}: we solve a problem similar to problem \eqref{prob:feasi_E} but without considering the phase errors, after which we apply the obtained solution to compute the actual achievable max-min harvested energy in the presence of the phase errors; 3) \textbf{Robust w/o time division}: the counterpart of the scheme in 1) without time division (i.e., with time-invariant transmit/reflect beamforming); 4) \textbf{Non-robust w/o time division}: the counterpart of the scheme in 2) without time division. 

In Fig. \ref{fig:E_vs_EUs}, we plot the achievable max-min harvested energy of the above schemes versus the number of EUs for $L = 5$. Firstly, it is observed that with increasing $J$, all the schemes experience a striking decrease in the max-min harvested energy. This is intuitive since the more the number of EUs, the more difficult it is to balance the energy fairness among different EUs. Secondly, we note that the time division-based schemes perform much better than their counterparts without time division. The reason is that in the considered overloaded system, the time division-based schemes allow the AP (IRS) to steer the energy (reflected) signals towards different EUs in different time slots, which improves the minimum harvested energy of more EUs (especially those with weak channel conditions). Lastly, it is expected that for both cases with and without time division, the non-robust design suffers a substantial performance loss compared to the robust one since the former does not considering the phase errors when designing the transmit/reflect beamforming and time allocation (if any). 
Nevertheless, ignoring the phase errors brings a more significant performance degradation to the time division-based scheme than to that without time division. This is because the phase errors have a greater negative impact on the former scheme adopting time-varying IRS beamforming than the latter one with time-invariant IRS beamforming. The above two observations demonstrate the importance of robust design for IRS-aided SWIPT systems with EH requirements and phase errors since a non-robust design can finally lead to an infeasible EH solution. 

\vspace{-1mm}
\subsection{Achievable Max-Min Throughput}
This subsection compares the achievable max-min throughputs of our proposed non-overlapping and overlapping UG schemes with those of the following two benchmark schemes: 1) \textbf{Random UG}: $a_{k,\ell}$ is non-optimized and randomly selected from $\{0,1\}$, $\forall k\in\mathcal K$, $\ell\in\mathcal L$;  
2) \textbf{Without UG}: the conventional IRS-aided SWIPT strategy as in \cite{2020_Qingqing_SWIPT_letter,2020_Qingqing_SWIPT_QoS,2020_Cunhua_SWIPT,2020_Wei_SWIPT_secure,2021_Shayan_SWIPT}, with the number of available time slots being $1$ (i.e., $\tau_1 = T$) and $a_{k,1} = 1$, $\forall k\in\mathcal K$. If any scheme is judged infeasible under certain setups, we assign a value of zero to its achievable max-min throughput as a means of factoring in the associated penalty. In addition, since the non-robust counterparts of the considered schemes almost always result in infeasible EH solutions, their simulation curves are omitted. \looseness=-1

\begin{figure}[!t]
	\vspace{-3mm}
	\hspace{-4.5mm}
	\subfigure[$E = 1\times 10^{-5}$ J.]{\label{fig:R_vs_IUs_E1}
		\includegraphics[scale=0.71]{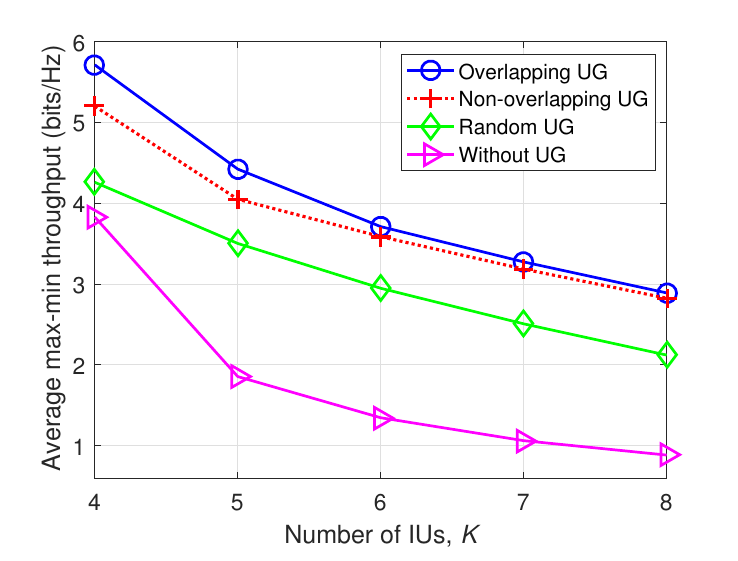}}
	\hspace{-8.5mm}
	\subfigure[$E = 2\times 10^{-5}$ J.]{\label{fig:R_vs_IUs_E2}
		\includegraphics[scale=0.71]{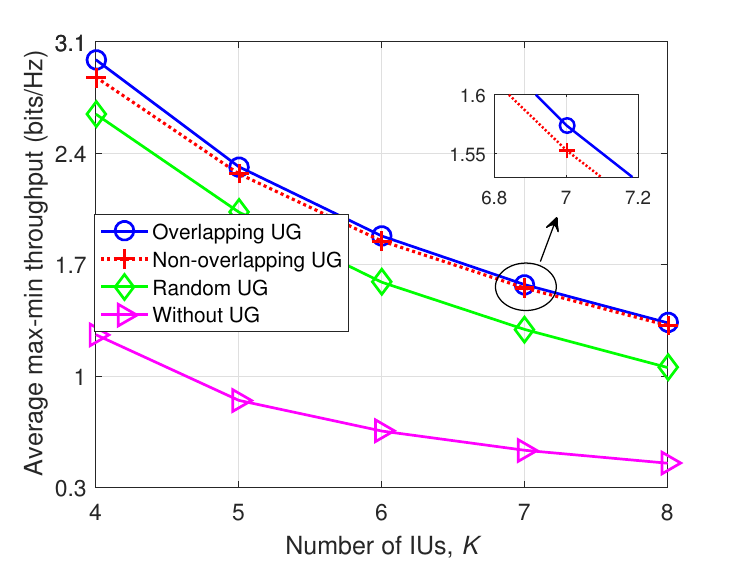}}\vspace{-1.4mm}
	\caption{Average max-min throughput versus the number of IUs for $J = 8$ and $L = 3$.}
	\label{fig:R_vs_IUs}
	\vspace{-4mm}
\end{figure}

\subsubsection{Impact of Number of IUs}
Fig. \ref{fig:R_vs_IUs} depicts the average max-min throughput versus the number of IUs when $E = 1 \times 10^{-5}$ and $2 \times 10^{-5}$ J, respectively. Here, we set $J = 8$ and $L = 3$. From Fig. \ref{fig:R_vs_IUs_E1}, it is first observed that the schemes adopting overlapping, non-overlapping, or random UG exhibit overwhelming superiority over that without UG, with the performance improvement in percentage increasing as $K$ increases. The reasons are twofold. For one thing, since the three UG-based schemes with time division make it easier to fulfill the EH constraints at the EUs (see Fig. \ref{fig:E_vs_EUs}), more degrees-of-freedom (DoF) are left for enhancing the performance of the IUs, as compared to the scheme without UG (and time division). For another, under the setting of $K + J > M$ and $K \geq M$, grouping the IUs can alleviate the inter-user interference more effectively than not grouping them, especially when $K$ is large. Second, the scheme with random UG performs not so well as those with optimized UG, which shows the importance of well-optimized UG for performance enhancement. Third, the overlapping UG scheme consistently outperforms its sub-scheme, i.e., the non-overlapping UG scheme, as the former enables more efficient utilization of all the available resources. However, it is noteworthy that as $K$ increases, the performance improvement of the overlapping UG scheme over the non-overlapping UG scheme becomes less pronounced. The explanation is that since increasing $K$ must lead to more severe inter-user interference, allowing some IUs to participate in multiple groups may no longer be a better choice or can only bring little throughput gain. 

From Fig. \ref{fig:R_vs_IUs_E2}, besides the
observations similar to those in Fig. \ref{fig:R_vs_IUs_E1}, we observe that the performance gap between the overlapping and non-overlapping UG schemes is marginal, even when $K$ is relatively small. This can be explained as follows. With $E = 2 \times 10^{-5}$ J, few resources are available for the IUs since the EUs with stringent EH requirements occupy most of them, which makes the overlapping UG scheme hardly bring its advantage in more efficient resource utilization for throughput improvement into play. 

\begin{figure}[!t]
	\hspace{-4.5mm}
	\subfigure[]{\label{fig:active_slots_vs_IUs}
		\includegraphics[scale=0.71]{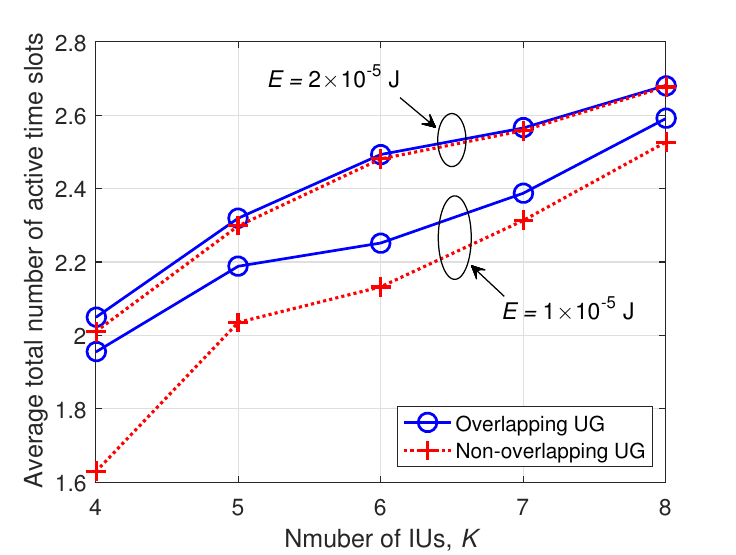}}
	\hspace{-8.5mm}
	\subfigure[]{\label{fig:sum_a_vs_IUs}
		\includegraphics[scale=0.71]{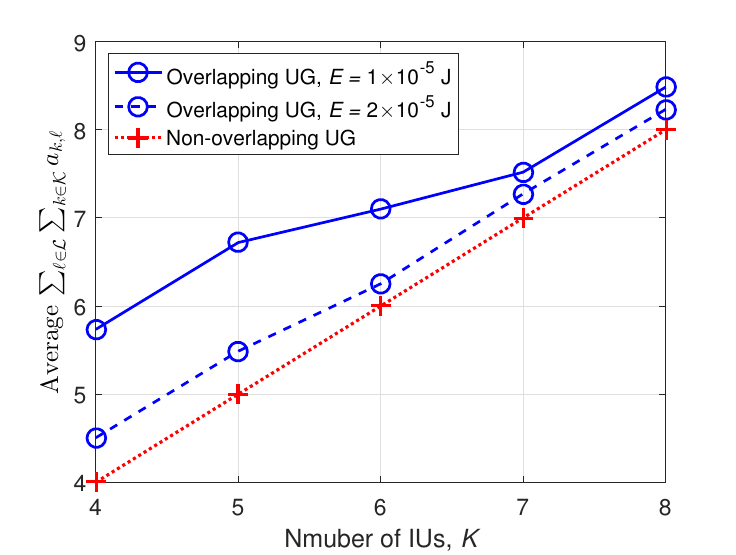}}\vspace{-1.4mm}
	\caption{(a) Average total number of active time slots and (b) average $\sum_{\ell\in\mathcal L}\sum_{k\in\mathcal K}a_{k,\ell}$ versus the number of IUs for $J = 8$ and $L = 3$.}
	\label{fig:resource_vs_IUs}
	\vspace{-2mm}
\end{figure} 

To gain more insights, we plot the corresponding average total number of active time slots and $\sum_{\ell\in\mathcal L}\sum_{k \in\mathcal K}a_{k,\ell}$ versus the number of IUs in Figs. \ref{fig:active_slots_vs_IUs} and \ref{fig:sum_a_vs_IUs}, respectively. Here, the active time slots refer to those with positive time durations, and the number of them also implies the number of IU groups. Fig. \ref{fig:active_slots_vs_IUs} shows that the average number of active time slots (as well as the average number of IU groups) increases with $K$, since the number of IUs that the AP can serve well in one time slot is limited. Moreover, it is worth pointing out that the number of IU groups is not the more the better, since the transmission duration allocated to each group is inversely proportional to the number of IU groups. This may explain why not all the available time slots are active. We also note that more active time slots are required for both the UG schemes when $E = 2 \times 10^{-5}$ J than when $E = 1 \times 10^{-5}$ J. Besides, it can be seen from Fig. \ref{fig:sum_a_vs_IUs} that for the overlapping UG scheme, the average $\sum_{\ell\in\mathcal L}\sum_{k \in\mathcal K}a_{k,\ell}$ is always larger than its corresponding $K$. This confirms that in certain channel realizations, some IUs do participate in more than one IU group. \looseness=-1

\begin{figure}[!t]
	\begin{minipage}[t]{0.5\linewidth}
		\hspace{-3mm}
		\includegraphics[scale=0.71]{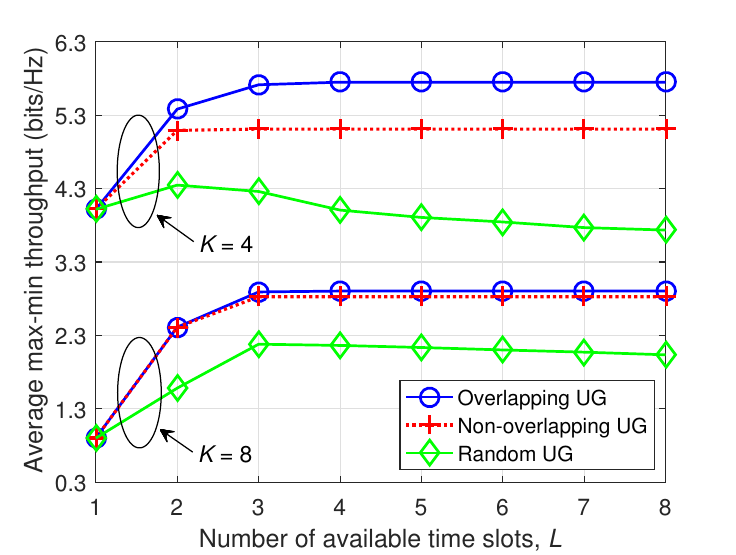}\vspace{-1.5mm}
		\caption{Average max-min throughput versus the number of\\ time slots for $J = 8$ and $E = 1\times10^{-5}$ J.}
		\label{fig:R_vs_L}
	\end{minipage}%
	\begin{minipage}[t]{0.5\linewidth}
		\centering
		\includegraphics[scale=0.71]{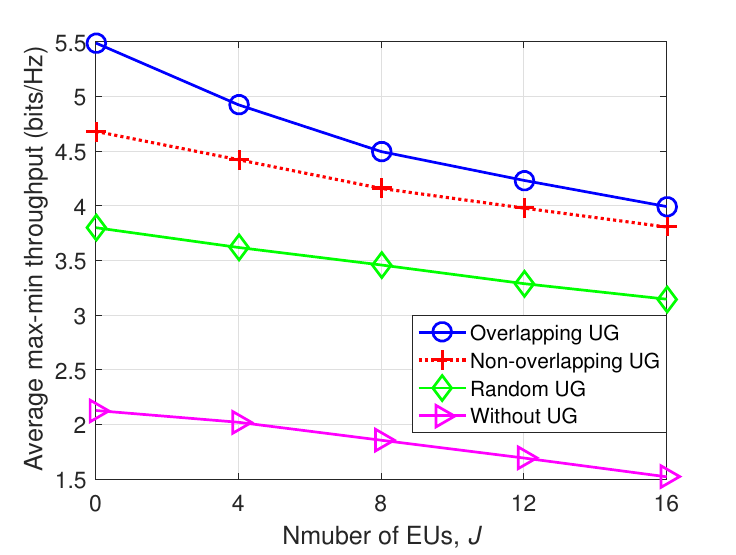}\vspace{-1.5mm}
		\caption{Average max-min throughput versus the number of EUs for $K = 5$, $L = 3$, and $E = 1\times10^{-5}$ J.}
		\label{fig:R_vs_EUs}
	\end{minipage}
	\vspace{-3mm}
\end{figure}

\subsubsection{Impact of Number of Available Time Slots}

In Fig. \ref{fig:R_vs_L}, we investigate the impact of the number of available time slots on the system performance for $J = 8$ and $E = 1 \times 10^{-5}$ J. It is observed that with increasing $L$, the max-min throughputs achieved by the overlapping and non-overlapping UG schemes first grow monotonically and then become gradually saturated. The reasons for this result are as follows. When $L$ is small, the increase in $L$ allows the formation of more IU groups, each with fewer IUs, enabling more spatial multiplexing gains to be achieved in each corresponding time slot. On the other hand, when $L$ is large enough, further increasing $L$ would no longer lead to an increased number of IU groups. This is because dividing the IUs into many more groups but each with a shorter transmission duration can be unfavorable for max-min throughput performance, which is confirmed by the trends of the curves representing the random UG scheme. 

\subsubsection{Impact of Number of EUs}
Fig. \ref{fig:R_vs_EUs} illustrates the average max-min throughput versus the number of EUs when $K = 5$, $L = 3$, and $E = 1\times10^{-5}$ J. As can be seen, the max-min throughputs achieved by all the schemes decrease rapidly with the increase of $J$. This is expected since the number of EH constraints increases with $J$, which narrows the feasible regions of the considered problems corresponding to these schemes. Additionally, in the absence of EUs (i.e., $J = 0$), the three schemes with UG still significantly outperform that without UG, thus further verifying the usefulness of grouping the IUs for max-min throughput improvement. Finally, the performance gap between the overlapping and non-overlapping UG schemes decreases as $J$ increases, which is consistent with the observation in Fig. \ref{fig:R_vs_IUs} that the increase in $E$ diminishes the advantage of the overlapping UG scheme over its non-overlapping counterpart.   

\begin{figure}[!t]
	\hspace{-4.5mm}
	\subfigure[$E = 1\times 10^{-5}$ J.]{\label{fig:R_vs_N_E1}
		\includegraphics[scale=0.71]{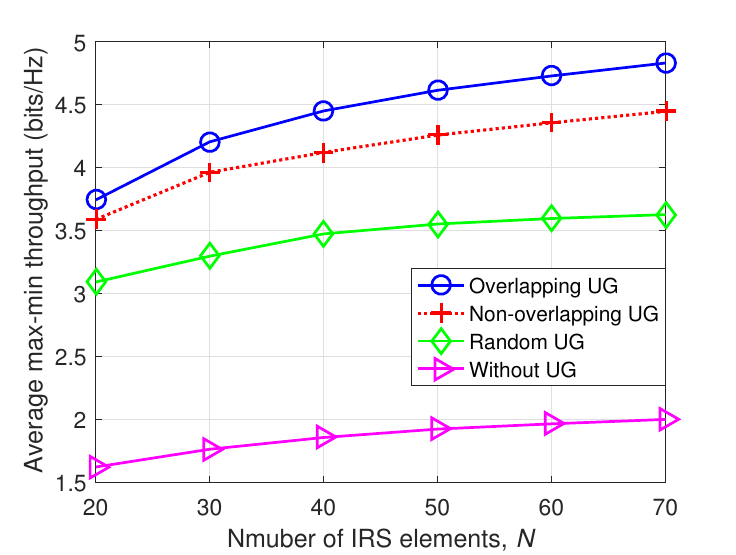}}
	\hspace{-8.5mm}
	\subfigure[$E = 2\times 10^{-5}$ J.]{\label{fig:R_vs_N_E2}
		\includegraphics[scale=0.71]{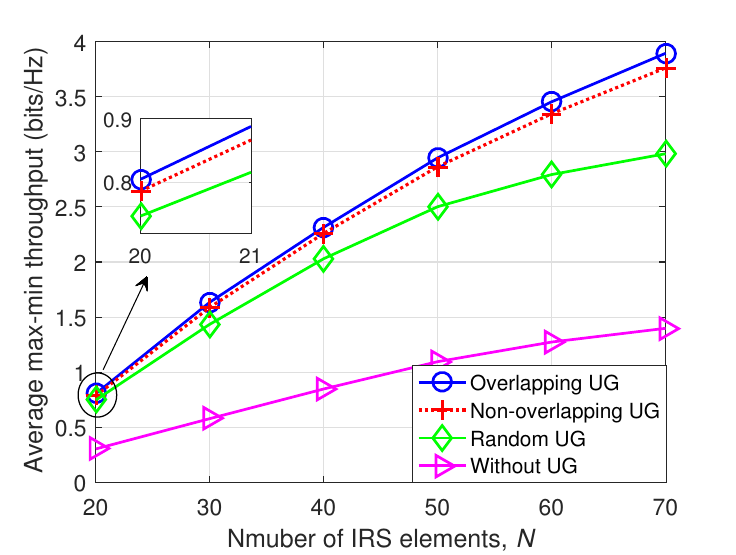}}\vspace{-1mm}
	\caption{Average max-min throughput versus the number of IUs for $K = 5$, $J = 8$ and $L = 3$.}
	\label{fig:R_vs_N}
	\vspace{-3mm}
\end{figure}

\subsubsection{Impact of Number of IRS Elements}
In Fig. \ref{fig:R_vs_N}, we plot the average max-min throughput versus the number of IRS elements when $E = 1 \times 10^{-5}$ and $2 \times 10^{-5}$ J, respectively. It is observed that the max-min throughputs achieved by all the schemes show upward trends as $N$ becomes larger, since more DoF are available for customizing more favorable channels. Nevertheless, the performance gains diminish with $N$, which is especially evident for the relatively smaller $E$. We explain this result based on the following two facts. First, increasing $N$ makes the EH requirement gradually less of a limiting factor to the performance. Second, the achievable max-min throughput of each scheme is upper-bounded by a finite value due to the AP's limited transmit power and transmission duration. Besides, we note that with the increase of $N$, the performance gap between the overlapping UG scheme and the other three schemes becomes more pronounced, since the former can better utilize the increased DoF. 

\section{Conclusion}\label{Sec:conclusion}
This paper considered an overloaded multiuser MISO downlink SWIPT system assisted by an IRS with phase errors. We grouped the IUs by considering two UG schemes, i.e., the non-overlapping UG scheme and the overlapping UG scheme, which do not allow and allow each IU to participate in multiple groups, respectively. Aiming to maximize the minimum throughput among all the IUs while satisfying the EH requirement of each EU, we formulated two design problems, each corresponding to one UG scheme, where the UG variables, the time allocation, and the transmit/reflect beamforming were jointly optimized. Computationally efficient algorithms were proposed to solve these two mixed-integer non-convex optimization problems suboptimally. Simulation results demonstrated that robust design is vital to practical IRS-aided SWIPT systems with phase errors since the solution obtained when ignoring the phase errors generally fails to satisfy the EH constraints.  
Moreover, our proposed UG schemes can remarkably improve the max-min throughput performance compared to the case without UG, as they enable higher active and passive beamforming gains by serving fewer IUs concurrently. 
Finally, unless the absolute difference between the number of transmit antennas at the AP and the number of IUs is small and the EH constraints are loose, the max-min throughput achieved by the non-overlapping UG scheme is comparable to that by the overlapping UG scheme. Thus, in most scenarios, the non-overlapping UG scheme is more attractive to practice systems due to its comparable performance to the overlapping UG scheme and extra advantage of easier implementation. 

\appendices
\section{Proof of Theorem \ref{theo}} \label{Appen_A}
According to \eqref{equ:SINR} and \eqref{equ:E}, $\mathbb E_{\tilde{\mathbf v}_\ell}\{ \gamma_{k,\ell}\}$ and $\mathbb E_{\tilde{\mathbf v}_\ell}\{ Q_j\}$ can be written in the following forms: 
{\addtolength{\jot}{5pt}
	\begin{align}
	\mathbb E_{\tilde{\mathbf v}_\ell}\{ \gamma_{k,\ell}\} & = \frac{a_{k,\ell}\mathbf w_{k,\ell}^H\mathbf H_k^H\mathcal P_\ell\mathbf H_k\mathbf w_{k,\ell}}{\sum_{i\in \mathcal K\backslash\{k\}}a_{i,\ell}\mathbf w_{i,\ell}^H\mathbf H_k^H\mathcal P_\ell\mathbf H_k\mathbf w_{i,\ell} + {\rm tr}\left(\mathbf H_k^H\mathcal P_\ell\mathbf H_k\mathbf W_{\mathrm E,\ell}\right)+ \sigma_k^2}, \  k\in\mathcal K, \ell\in\mathcal L, \label{expecta_SINR_orig}\\ 
	\mathbb E_{\tilde{\mathbf v}_\ell}\{ Q_j\} & = \sum_{\ell\in\mathcal L}\tau_\ell\left(\sum_{k\in\mathcal K}a_{k,\ell}\mathbf w_{k,\ell}^H\mathbf G_j^H\mathcal P_\ell\mathbf G_j\mathbf w_{k,\ell} + {\rm tr}\left(\mathbf G_j^H\mathcal P_\ell\mathbf G_j\mathbf W_{\mathrm E, \ell} \right) \right), \  j\in\mathcal J, \label{expecta_E_orig}
	\end{align}}%
where $\mathcal P_\ell = \mathbb E_{\tilde{\mathbf v}_\ell}\left\lbrace \left(\mathbf v_\ell\odot\tilde {\mathbf v}_\ell\right)\left(\mathbf v_\ell\odot\tilde {\mathbf v}_\ell\right)^H \right\rbrace$. Then, the problem of deriving the closed-form expressions of $\mathbb E_{\tilde{\mathbf v}_\ell}\{ \gamma_{k,\ell}\}$ and $\mathbb E_{\tilde{\mathbf v}_\ell}\{ Q_j\}$ is converted into that for $\mathcal P_\ell$. Notice that $\mathcal P_\ell$ can be recast as $\mathcal P_\ell = {\rm diag}\left(\mathbf v_\ell\right)\mathbb E_{\tilde{\mathbf v}_\ell}\left\lbrace\tilde{\mathbf v}_\ell\tilde{\mathbf v}_{\ell}^H\right\rbrace {\rm diag}\left(\mathbf v_{\ell}^H\right)$ with the expression of $\mathbb E_{\tilde{\mathbf v}_\ell}\left\lbrace\tilde{\mathbf v}_\ell\tilde{\mathbf v}_{\ell}^H\right\rbrace$ given by
\begin{align}\label{equ:matrix}
\begin{bmatrix}
1& \mathbb E_{\Delta\tilde{\theta}_{\ell,2,1}}\left\lbrace e^{\jmath\Delta\tilde{\theta}_{\ell,2,1}}\right\rbrace & \cdots & \mathbb E_{\Delta\tilde{\theta}_{\ell,N,1}}\left\lbrace e^{\jmath\Delta\tilde{\theta}_{\ell,N,1}}\right\rbrace  & \mathbb E_{\tilde\tilde{\theta}_{\ell,1}}\left\lbrace e^{-\jmath\tilde{\theta}_{\ell,1}}\right\rbrace \\
\mathbb E_{\Delta\tilde{\theta}_{\ell,1,2}}\left\lbrace e^{\jmath\Delta\tilde{\theta}_{\ell,1,2}}\right\rbrace & 1 & \cdots & \mathbb E_{\Delta\tilde{\theta}_{\ell,N,2}}\left\lbrace e^{\jmath\Delta\tilde{\theta}_{\ell,N,2}}\right\rbrace & \mathbb E_{\tilde\theta_{\ell,2}}\left\lbrace e^{-\jmath\tilde{\theta}_{\ell,2}}\right\rbrace \\
\vdots & \vdots & \ddots & \vdots & \vdots \\
\mathbb E_{\Delta\tilde{\theta}_{\ell,1,N}}\left\lbrace e^{\jmath\Delta\tilde{\theta}_{\ell,1,N}}\right\rbrace & \mathbb E_{\Delta\tilde{\theta}_{\ell,2,N}}\left\lbrace e^{\jmath\Delta\tilde{\theta}_{\ell,2,N}}\right\rbrace & \cdots &  1 & \mathbb E_{\tilde\theta_{\ell,N}}\left\lbrace e^{-\jmath\tilde{\theta}_{\ell,N}}\right\rbrace \\
\mathbb E_{\tilde\theta_{\ell,1}}\left\lbrace e^{\jmath\tilde{\theta}_{\ell,1}}\right\rbrace  & \mathbb E_{\tilde\theta_{\ell,2}}\left\lbrace e^{\jmath\tilde{\theta}_{\ell,2}}\right\rbrace  & \cdots & \mathbb E_{\tilde\theta_{\ell,N}}\left\lbrace e^{\jmath\tilde{\theta}_{\ell,N}}\right\rbrace & 1 \\
\end{bmatrix}. 
\end{align}
In \eqref{equ:matrix}, $\Delta\tilde{\theta}_{\ell,m,n} \triangleq \tilde{\theta}_{\ell,m} - \tilde{\theta}_{\ell,n}$, $m,n\in\mathcal N$, $m\neq n$, $\ell\in\mathcal L$. Since $\tilde{\theta}_{\ell,m}$ and $\tilde{\theta}_{\ell,n}$ are uniformly distributed on $\left[-{\pi}/{2}, {\pi}/{2}\right]$, $\Delta\tilde{\theta}_{\ell,m,n}$ follows a triangular distribution on $[-\pi,\pi]$ and its probability density function can be expressed as
\begin{align}\label{PDF}
f\left(\Delta\tilde{\theta}_{\ell,m,n}\right) = \begin{cases}
\frac{\Delta\tilde{\theta}_{\ell,m,n}}{\pi^2} + \frac{1}{\pi}, & \Delta\tilde{\theta}_{\ell,m,n} \in [-\pi,0],\\
-\frac{\Delta\tilde{\theta}_{\ell,m,n}}{\pi^2} + \frac{1}{\pi}, & \Delta\tilde{\theta}_{\ell,m,n} \in (0,\pi], \\
0, & \text{otherwise}. 
\end{cases}
\end{align}
With \eqref{PDF}, we have 
\begin{align}\label{expectation_1}
E_{\Delta\tilde{\theta}_{\ell,m,n}}\left\lbrace e^{\jmath\Delta\tilde{\theta}_{\ell,m,n}}\right\rbrace = & \int_{-\pi}^{0}\left(\frac{\Delta\tilde{\theta}_{\ell,m,n}}{\pi^2} + \frac{1}{\pi} \right)e^{\jmath\Delta\tilde{\theta}_{\ell,m,n}}d\Delta\tilde{\theta}_{\ell,m,n} \nonumber\\
& + \int_{0}^{\pi}\left(-\frac{\Delta\tilde{\theta}_{\ell,m,n}}{\pi^2} + \frac{1}{\pi} \right)e^{\jmath\Delta\tilde{\theta}_{\ell,m,n}}d\Delta\tilde{\theta}_{\ell,m,n} = \frac{4}{\pi^2}. 
\end{align}
On the other hand, since $\tilde{\theta}_{\ell,n}$ obeys a uniform distribution on $[-\pi/2,\pi/2]$, one can easily derive that 
\begin{align}
\mathbb E_{\tilde\theta_{\ell,N}}\left\lbrace e^{\jmath\tilde{\theta}_{\ell,N}}\right\rbrace = \int_{-\frac{\pi}{2}}^{\frac{\pi}{2}}\frac{1}{\pi}e^{\jmath\tilde{\theta}_{\ell,N}}d\tilde{\theta}_{\ell,N} = \frac{2}{\pi}, \label{expectation_2}\\
\mathbb E_{\tilde\theta_{\ell,N}}\left\lbrace e^{-\jmath\tilde{\theta}_{\ell,N}}\right\rbrace = \int_{-\frac{\pi}{2}}^{\frac{\pi}{2}}\frac{1}{\pi}e^{-\jmath\tilde{\theta}_{\ell,N}}d\tilde{\theta}_{\ell,N} = \frac{2}{\pi}. \label{expectation_3} 
\end{align}
By substituting \eqref{expectation_1}-\eqref{expectation_3} into \eqref{equ:matrix}, we have
\begin{align}
\mathbb E_{\tilde{\mathbf v}_\ell}\left\lbrace\tilde{\mathbf v}_\ell\tilde{\mathbf v}_{\ell}^H\right\rbrace = \begin{bmatrix}
1& \frac{4}{\pi^2} & \cdots & \frac{4}{\pi^2} & \frac{2}{\pi} \\
\frac{4}{\pi^2} & 1 & \cdots & \frac{4}{\pi^2}  & \frac{2}{\pi} \\
\vdots & \vdots & \ddots & \vdots & \vdots \\
\frac{4}{\pi^2} & \frac{4}{\pi^2} & \cdots &  1 & \frac{2}{\pi} \\
\frac{2}{\pi} & \frac{2}{\pi} & \cdots & \frac{2}{\pi} & 1 \\
\end{bmatrix} = \mathbf Z,
\end{align}
with which, the closed-form expression of $\mathcal P_\ell$ is given by $\mathcal P_\ell = {\rm diag}\left(\mathbf v_\ell\right)\mathbf Z{\rm diag}\left(\mathbf v_{\ell}^H\right)$. Finally, by replacing $\mathcal P_\ell$ in \eqref{expecta_SINR_orig} and \eqref{expecta_E_orig} with its closed-form expression, we arrive at \eqref{expecta_SINR} and \eqref{expecta_E}, respectively. This completes the proof of Theorem \ref{theo}.  

\vspace{1mm}
\section{Proof of Lemma \ref{lem}} \label{Appen_B}
We prove Lemma \ref{lem} by showing that ${\rm tr}\left(\mathbf Y_{j,\ell}\mathbf W_{\mathrm E,\ell} \right) = \mathbf v_{\ell}^H\mathbf Q_{j,\mathrm E, \ell}\mathbf v_{\ell}$, $\forall \ell\in\mathcal L'$. First, $\{q_{\ell,m}\}$ and $\{\mathbf w_{\rm E,\ell,m}\}$ can be obtained from the eigenvalue decomposition of $\mathbf W_{\rm E, \ell}$ and we can express $\mathbf W_{\rm E, \ell}$ as $\mathbf W_{\rm E,\ell} = \sum_{m=1}^{r_{\rm E,\ell}}q_{\ell,m}\mathbf w_{\rm E,\ell,m}\mathbf w_{\rm E,\ell,m}^H$. Second, recall that $\mathbf Y_{j,\ell} = \mathbf G_j^H{\rm diag}\left(\mathbf v_\ell\right)\mathbf Z{\rm diag}\left(\mathbf v_{\ell}^H\right)\mathbf G_j$, then we can derive that
\begin{align}\label{appen_B_equ}
{\rm tr}\left(\mathbf Y_{j,\ell} \mathbf W_{\rm E,\ell}\right) & \overset{(a)}{=} \sum_{m=1}^{r_{\rm E,\ell}}q_{\ell,m}\mathbf w_{\rm E,\ell,m}^H\mathbf G_j^H{\rm diag}\left(\mathbf v_\ell\right)\mathbf Z{\rm diag}\left(\mathbf v_{\ell}^H\right)\mathbf G_j\mathbf w_{\rm E,\ell,m}\nonumber\\
& \overset{(b)}{=} \sum_{m=1}^{r_{\rm E,\ell}}q_{\ell,m}{\rm conj}\left(\mathbf v_{\ell}^H{\rm diag}\left(\mathbf G_j\mathbf w_{\rm E,\ell,m}\right)\mathbf Z\left( {\rm diag}\left(\mathbf G_j\mathbf w_{\rm E,\ell,m}\right)\right)^H\mathbf v_{\ell}\right) \nonumber\\
& \overset{(c)}{=} \sum_{m=1}^{r_{\rm E,\ell}}q_{\ell,m}\mathbf v_{\ell}^H{\rm diag}\left(\mathbf G_j\mathbf w_{\rm E,\ell,m}\right)\mathbf Z\left( {\rm diag}\left(\mathbf G_j\mathbf w_{\rm E,\ell,m}\right)\right)^H\mathbf v_{\ell} \nonumber\\
& = \mathbf v_{\ell}^H\left(\sum_{m=1}^{r_{\rm E,\ell}}q_{\ell,m}{\rm diag}\left(\mathbf G_j\mathbf w_{\rm E,\ell,m}\right)\mathbf Z\left( {\rm diag}\left(\mathbf G_j\mathbf w_{\rm E,\ell,m}\right)\right)^H\right)\mathbf v_{\ell} = \mathbf v_{\ell}^H\mathbf Q_{j,\mathrm E, \ell}\mathbf v_{\ell},
\end{align}
where the equality $(a)$ utilizes the properties of the trace operator, the equality $(b)$ holds due to the facts that $\mathbf w_{\rm E,\ell,m}^H\mathbf G_j^H{\rm diag}\left(\mathbf v_\ell\right) = {\rm conj}\left(\mathbf v_{\ell}^H{\rm diag}\left(\mathbf G_j\mathbf w_{\rm E,\ell,m}\right)\right)$ and ${\rm diag}\left(\mathbf v_{\ell}^H\right)\mathbf G_j\mathbf w_{\rm E,\ell,m} = {\rm conj}\left( \left( {\rm diag}\left(\mathbf G_j\mathbf w_{\rm E,\ell,m}\right)\right)^H\mathbf v_{\ell}\right)$, and the equality $(c)$ is true since each term in the left-hand-side of $(c)$ is a real number. With \eqref{appen_B_equ}, we can readily verify that constraint \eqref{E_cons:b} is equivalent to constraint \eqref{E_cons:b_eqv}, which completes the proof of Lemma \ref{lem}. 

\vspace{1mm}
\bibliographystyle{IEEEtran}
\bibliography{ref}

\end{sloppypar}
\end{document}